\documentclass[sigconf, nonacm]{acmart}
\usepackage[algo2e,ruled,vlined,linesnumbered]{algorithm2e}
\usepackage{url}
\usepackage{enumitem}
\usepackage{multirow} 
\usepackage{subcaption}

\newcommand{\journey}{\text{journey}}
\newcommand{\eat}{\text{eat}}
\newcommand{\arr}{\text{arr}}
\newcommand{\dep}{\text{dep}}

\newcommand{\leftvertex}{\text{left}}
\newcommand{\rightvertex}{\text{right}}

\newcommand\vldbdoi{XX.XX/XXX.XX}

\newcommand\vldbavailabilityurl{https://github.com/mithinti/EAT-FPD-Algorithms}
\newcommand\vldbpagestyle{plain}

\begin{document}
\title{Efficient Algorithms for Earliest and Fastest Paths in Public Transport Networks}

%%
%% The "author" command and its associated commands are used to define the authors and their affiliations.
\author{Mithinti Srikanth}
\affiliation{%
  \institution{Indian Institute of Technology Tirupati}
  \city{Tirupati}
  \state{Andhra Pradesh}
  \country{India}
}
\email{srikanth.mithinti@gmail.com}

\author{G. Ramakrishna}
\orcid{0000-0002-1825-0097}
\affiliation{%
  \institution{Indian Institute of Technology Tirupati}
  \city{Tirupati}
  \state{Andhra Pradesh}
  \country{India}
}
\email{rama@iittp.ac.in}

\begin{abstract}
Public transport administrators rely on efficient algorithms for various problems that arise in public transport networks. In particular, our study focused on designing linear-time algorithms for two fundamental path problems: the earliest arrival time (\textsc{eat}) and the fastest path duration (\textsc{fpd})  on public transportation data.
%by proposing structural properties of these two problems. 
We conduct a comparative analysis with state-of-the-art algorithms. The results are quite promising, indicating substantial efficiency improvements. Specifically, the fastest path problem shows a remarkable 34-fold speedup, while the earliest arrival time problem exhibits an even more impressive 183-fold speedup. These findings highlight the effectiveness of our algorithms to solve \textsc{eat} and \textsc{fpd} problems in public transport, and eventually help public administrators to enrich the urban transport experience.
%to handle relevant problems.
%in optimizing path-finding processes within public transportation networks represented as temporal graphs. 
\end{abstract}

\maketitle

\pagestyle{\vldbpagestyle}
% \begingroup\small\noindent\raggedright\textbf{PVLDB Reference Format:}\\
% \vldbauthors. \vldbtitle. PVLDB, \vldbvolume(\vldbissue): \vldbpages, \vldbyear.\\
% \href{https://doi.org/\vldbdoi}{doi:\vldbdoi}
% \endgroup
% \begingroup
% \renewcommand\thefootnote{}\footnote{\noindent
% This work is licensed under the Creative Commons BY-NC-ND 4.0 International License. Visit \url{https://creativecommons.org/licenses/by-nc-nd/4.0/} to view a copy of this license. For any use beyond those covered by this license, obtain permission by emailing \href{mailto:info@vldb.org}{info@vldb.org}. Copyright is held by the owner/author(s). Publication rights licensed to the VLDB Endowment. \\
% \raggedright Proceedings of the VLDB Endowment, Vol. \vldbvolume, No. \vldbissue\ %
% ISSN 2150-8097. \\
% \href{https://doi.org/\vldbdoi}{doi:\vldbdoi} \\
% }\addtocounter{footnote}{-1}\endgroup

% \ifdefempty{\vldbavailabilityurl}{}{
% \vspace{.3cm}
% \begingroup\small\noindent\raggedright\textbf{PVLDB Artifact Availability:}\\
% The source code, data, and/or other artifacts have been made available at \url{\vldbavailabilityurl}.
% \endgroup


\section{Introduction}

Route optimization on public transportation is crucial for urban public transport administrators involved in route planning. 
Intelligent route optimization algorithms are being used in transportation software, to enhance traffic flow and reduce environmental impact. These algorithms are used by various stakeholders to address critical questions such as travel time estimation, maximizing city coverage, and efficient urban transport experience. A vital component of public transportation networks is the scheduled timetable information, encompassing vehicle departure and arrival times at stops across various routes. This valuable data set is efficiently handled and represented using temporal graphs.
% A public transport network can be modelled using temporal graphs.

A \emph{temporal graph} is a weighted directed graph in which departure time and duration time are assigned to each edge. Given a temporal graph $G$, a source vertex $s$, and a ready time at source vertex $rt$, the earliest arrival time problem (\textsc{eat}) is to find the minimum arrival times of paths from $s$ to all the vertices, in which the departure of each path is at least $rt$.
Given a temporal graph $G$ and a source vertex $s$, the fastest path duration problem (\textsc{fpd}) is to find minimum journey times from $s$ to all the vertices, where the journey time of a path is the difference between its arrival time and the departure time. 

% \textcolor{red}{Formal defintions of fastest path duration, earliest arrival time}
% \textcolor{red}{computing the fastest paths in temporal graph is challenging as sub path of a fastest path may not be fastest}

% -chain edges is 54 percent of the primary edges in TRG.
% - TRG has 1.5 times edges from the original temporal graph.
% - ESD-graph has 60 percent of the edges with respect to TRG.
% - ESD-graph has 1.04 times edges from the original temporal graph.
% - In public transportation networks, the chain length tends to be slightly higher than the temporal out-degree at every vertex u. Notably, the temporal average out-degree in public transportation networks is significantly larger than general temporal graphs.

% Various graph representations are being evolved to design efficient graph algorithms based on various factors such as the nature of the graph problem and the corresponding application, type of computations and memory access patterns. 
% Broadly, various path problems can be classified into two variants namely single-source and goal-oriented. In the single-source variant, the objective is to compute certain values such as distance, earliest arrival time, fastest path duration etc., from the given source to all the vertices. On the other hand, any interesting measure of paths is computed from the given source to the given target vertex in the goal-oriented variant. We consider single-source \textsc{eat} and single-source \textsc{fdp} on public transport networks in this paper.

Multiple graph representations are evolved based on the nature of the graph problems, associated applications, and various factors such as computation type and memory access patterns. Broadly, various path problems can be classified into two variants namely single-source and goal-oriented. 
% In the single-source variant, the objective is to compute certain values such as distance, earliest arrival time, fastest path duration etc., from the given source to all the vertices.
Single-source problems aim to find values like distance or arrival time from one starting point to all vertices. Goal-oriented problems calculate path measures from a specific source to a particular target. This paper focuses on single-source \textsc{eat} and single-source \textsc{fpd} in public transport networks.
% Path problems can be grouped into two main types: single-source and goal-oriented. Single-source problems aim to find values like distance or arrival time from one starting point to all vertices. Goal-oriented problems calculate path measures from a specific source to a particular target. This paper focuses on single-source \textsc{eat} and single-source \textsc{fdp} in public transport networks.

% To the best of our knowledge, 
To date, edge stream algorithm is the latest algorithm to solve single-source \textsc{eat} problem  using edge-stream representation in public transport networks and real world temporal graphs ~\cite{Dibbelt13intriguinglysimple} ~\cite{wu_2016_Fast_EAT_Temporal}.
Single-pass and traversal based multi-pass algorithms are the best known algorithms to solve single-source \textsc{fpd} problem in real word temporal graphs, by using edge-stream  and time-respecting graph (\textsc{trg}) representations, respectively ~\cite{wu_2016_Fast_EAT_Temporal, Sahani2021}.

In the edge-stream format all the edges of a graph are arranged in an array in non-decreasing order based on their departure time. 
Algorithms designed on the edge-stream format to solve \textsc{eat} and \textsc{fpd} perform well due to spatial data locality. The drawback in these algorithms (\cite{Dibbelt13intriguinglysimple, wu_2016_Fast_EAT_Temporal}) is that all edges of the graph are processed independent of the source vertex given in the query time.

The graph traversal based multi-pass \textsc{trg} algorithm performs better than the single-pass algorithm to solve \textsc{fpd}, due to pruning. In the \textsc{trg} algorithm, whenever a node $(u,t)$ is visited, all the other nodes  $(u, t')$, where $t' > t$ are visited, and all of their outgoing edges are processed. We observe the following two drawbacks in the \textsc{trg} approach \cite{Sahani2021}.  
i) For every vertex $u$ in $G$, all the departure nodes $(u,t)$ are connected by a chain of edges, and the chain length is at least the temporal out-degree of $u$.  In public transport networks, the temporal out-degree is quite large when compared to their static out-degree as shown in Table~\ref{tab:AvgOutDegree}. From Fig~\ref{fig:chainEdgesprocessingpercentage}, it is evident that around 45\% of the total running time is spent towards processing chain edges, which is a bottleneck.
%All the outgoing edges of visited nodes are processed in the algorithm to propagate the starting time of a journey, whose final effect is due to one temporal edge for all the temporal edges associated with a static edge. 
ii) For every static edge $(u,v)$, all the temporal edges $(u,v,t,\lambda)$ are processed if $u$ is reached on or before $t$, whose final effect is due to one temporal edge. In other words, whenever a vertex $u$ is visited at time $t$, processing $(u,v, t', \lambda')$ is not required, if there exist another edge $(u,v, t'', \lambda'')$ such that $t' + \lambda' > t'' + \lambda''$ or $t' < t$. 

\begin{table}[H]
\centering
\resizebox{\columnwidth}{!}{
\begin{tabular}{l|cc|cc|}
\cline{2-5}
 & \multicolumn{2}{c|}{\textbf{Temporal Out-Degree}} & \multicolumn{2}{c|}{\textbf{Static Out-Degree}} \\ \hline
\multicolumn{1}{|c|}{\textbf{Data Sets}} & \multicolumn{1}{c|}{\textbf{\begin{tabular}[c]{@{}c@{}}Max \\ Out-Degree\end{tabular}}} & \textbf{\begin{tabular}[c]{@{}c@{}}Average \\ Out-Degree\end{tabular}} & \multicolumn{1}{c|}{\textbf{\begin{tabular}[c]{@{}c@{}}Max \\ Out-Degree\end{tabular}}} & \textbf{\begin{tabular}[c]{@{}c@{}}Average \\ Out-Degree\end{tabular}} \\ \hline
\multicolumn{1}{|l|}{\textbf{Chicago}} & \multicolumn{1}{c|}{1315} & 408 & \multicolumn{1}{c|}{17} & 3 \\ \hline
\multicolumn{1}{|l|}{\textbf{London}} & \multicolumn{1}{c|}{7948} & 675 & \multicolumn{1}{c|}{7} & 1.3 \\ \hline
\multicolumn{1}{|l|}{\textbf{Los Angels}} & \multicolumn{1}{c|}{2069} & 142 & \multicolumn{1}{c|}{7} & 1.3 \\ \hline
\multicolumn{1}{|l|}{\textbf{Madrid}} & \multicolumn{1}{c|}{2940} & 425 & \multicolumn{1}{c|}{8} & 1.5 \\ \hline
\multicolumn{1}{|l|}{\textbf{Newyork}} & \multicolumn{1}{c|}{1480} & 521 & \multicolumn{1}{c|}{3} & 1.2 \\ \hline
\multicolumn{1}{|l|}{\textbf{Paris}} & \multicolumn{1}{c|}{83209} & 2599 & \multicolumn{1}{c|}{61} & 3 \\ \hline
\multicolumn{1}{|l|}{\textbf{Petersburg}} & \multicolumn{1}{c|}{11402} & 586 & \multicolumn{1}{c|}{22} & 1.5 \\ \hline
\multicolumn{1}{|l|}{\textbf{Sweden}} & \multicolumn{1}{c|}{17740} & 144 & \multicolumn{1}{c|}{43} & 2 \\ \hline
\multicolumn{1}{|l|}{\textbf{Switzerland}} & \multicolumn{1}{c|}{28315} & 310 & \multicolumn{1}{c|}{49} & 2 \\ \hline
\end{tabular}
}
\caption{\centering Temporal degree and static degree of public transport networks}
\label{tab:AvgOutDegree}
\end{table}

\begin{figure}[H]
\centering
\includegraphics[width=\columnwidth]{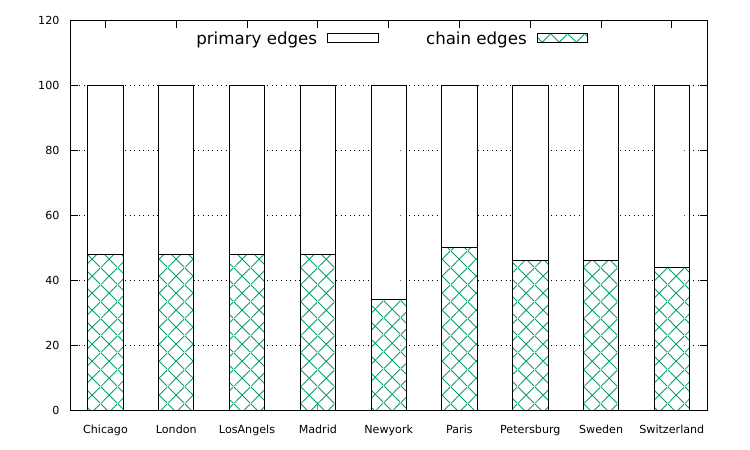}
\caption{Time spent on chain edges processing in \textsc{trg} based \textsc{eat} algorithm}
\Description{Temporal Graph}
\label{fig:chainEdgesprocessingpercentage}
\end{figure}

This motivates us to define certain dependencies between temporal edges and process only necessary dependencies to reduce the computations. Fortunately, only such important dependencies are captured in the edge-scan-dependency (\textsc{esd}) graph \cite{ESDG}. However, in their approach, edges are arranged in a special order based on the associated \textsc{esd}-graph and ignore the topological structure, to solve \textsc{eat} problem. We explore structural properties of \textsc{esd}-graph and use the topology of the \textsc{esd}-graph to solve \textsc{eat} and \textsc{fpd} problem, to overcome the drawbacks identified, and beat the existing results. Our key contributions are outlined as follows.

%This paper tackles two path problems: the earliest arrival time problem and the fastest path problem, formally defined below. These solutions aid public administrators in addressing route optimization and associated issues.Sanaz Gheibi et al., \cite{Sahani2021} introduced a memory-efficient data structure known as the Time Respecting Graph (TRG) to represent a temporal graph and proposes an algorithm to compute the fastest path duration from the given source vertex $s$ to the rest of the vertices present in the temporal graph $G$ using transformed graph $G'$ designed by them. Wu et.al.,\cite{wu2016FastEAT} proposed edge stream based algorithms to compute earliest arrival time and fastest path duration from the given source vertex to the rest of the vertices present in the graph.

%\textbf{Assumptions (if any)} non-negative integer edge weights, No delays, delays or uncertain events

% \begin{itemize} [label={}, left=0pt, itemsep=0.5pt]
\begin{itemize} [left=0pt, itemsep=0.5pt]

\item \textit{Characterization:} We introduce the notion of useful dominating paths and show that there is a one to one mapping between useful-dominating-paths in a temporal graph $G$ and paths in the transformed graph $\Tilde{G}$ of $G$. This helps to eliminate traversal of many unnecessary paths in our algorithms. %Also,  time validation is not required during the traversal of $\Tilde{G}$, as every path in $\Tilde{G}$ corresponds to a time respecting path in $G$.

\item \textit{Data-Structures:} We use the topology of edge-scan-dependency graph data structure to avoid the time validation computation, during the traversal of $\Tilde{G}$, as every path in $\Tilde{G}$ corresponds to a time respecting path in $G$.
%We utilized the Edge Scan Dependency Graph data structure to store public transportation data and extensively leveraged its features. 

\item \textit{Algorithms:} For a temporal graph on $n$ vertices and $m$ edges, we devise $O(m+n)$ algorithms to solve the fastest path duration and earliest arrival time problems. In our algorithms, we make sure that every edge is processed at most once. % from the given source vertex to the rest of the vertices in the temporal graph of type public transportation.

\item \textit{Speedup:} We run our algorithms on the nine real-time public transportation data sets. In practice, we avoid processing many edges.  Thus, the fastest path duration algorithm obtains a 34-fold speed up and the earliest arrival path duration algorithm,  obtains a 183-fold speed up, compared with the state-of-the-art algorithms.
\end{itemize}

% \textcolor{blue}{Novelty}: identify right property.... choose right data structure, right algorithm, 

% counting temporal paths is #P-hard \cite{NumberofPaths_EnrightMM23}
\section{Preliminaries}

A temporal graph is a weighted directed graph, in which multiple edges exist between the same pair of vertices, and the edges associate with time information.
We use $G$ to denote a temporal graph. Let $V(G)$ and $E(G)$ denote the set of vertices and set of edges, respectively in $G$.
An edge in a temporal graph $G$ is represented by a 4-tuple $(u,v,t,\lambda)$, where $u$ and $v$ are the end vertices of the edge, $t$ and $\lambda$ are positive integers, $t$ denotes the departure time at $u$, and $\lambda$ denotes the duration time from $u$ to $v$; $t+\lambda$ is considered as arrival time at $v$. 
A sequence $(e_1, e_2, \ldots, e_k)$ of edges in $G$ is a \emph{time respecting path}, if  it joins a sequence of vertices  and the departure time of every edge is at least the arrival time of the previous edge.
Let $P$ be a time respecting in $G$. Then, $\dep(P)$ and $\arr(P)$ denote the departure and arrival times, respectively. 
Further, the journey time of $P$ is defined as $\arr(P) - \dep(P)$.
Let $s$ and $z$ be two vertices in $G$. A time-respecting path from $s$ whose departure time is at least the given ready time $rt$, and the arrival time at $z$ is minimum is referred to as an \emph{earliest arrival path}. Similarly, a time-respecting path from $s$ to $z$ is a fastest path if its journey time is minimum over all the paths from $s$ to $z$. The arrival time of the earliest arrival path and the journey time of the fastest path are known as \emph{earliest arrival time} and \emph{fastest path duration}. In Figure~\ref{fig:pathexamplelong}, for the given ready time $3$, the sequence of (3,1), (5,4),(9,3) edges forms an earliest arrival path from $v_1$ to $v_9$, whose arrival time is $12$ and journey time is $9$. Further, the sequence of (9,3),(12,2),(14,1) edges is a fastest  path from $v_1$ to $v_9$, because its journey time is $6$. In the \textsc{eat} and \textsc{fpd} problems, the goal is to compute the earliest arrival times and fastest path durations from a source vertex to all the vertices.

% \textcolor{red}{Frequently used Notations table}

% For a temporal graph $G$, we use $\kappa(G)$ to denote the maximum number of incoming edges to $e=(u,v,t,\lambda)$ with distinct arrival times that lie between $t$ and $t+\lambda$ over all the temporal edges $e$ in $G$. A path $P$ from $u$ to $v$ in a temporal graph is a \emph{fastest path}, if the journey time of $P$ is minimum over all  paths from $u$ to $v$. The journey time of a path $P$ is defined as the sum of the total travelling time and total waiting time involved in $P$. 

\begin{figure}
\centering
\includegraphics[width=\columnwidth]{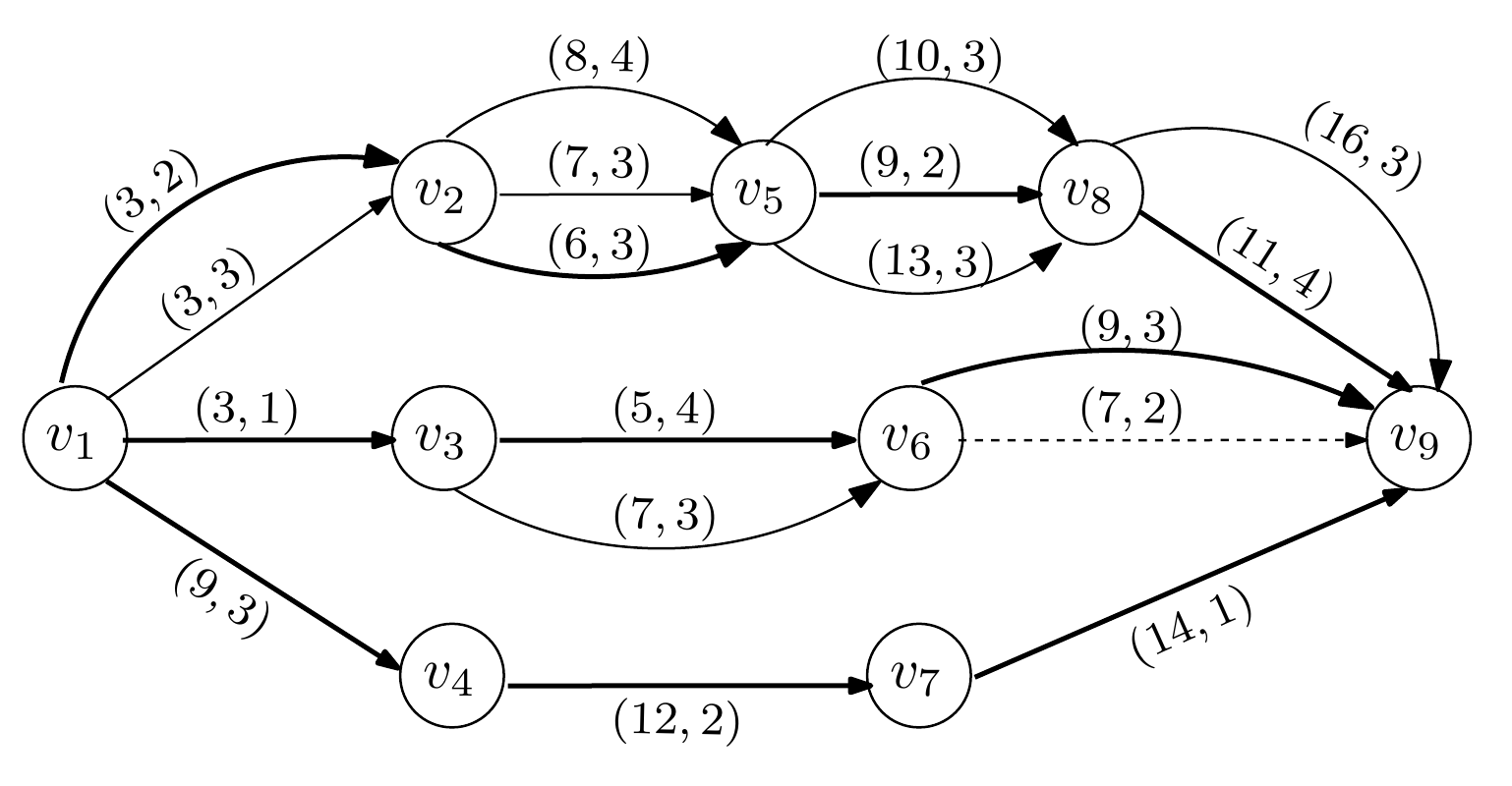}
\caption{A temporal graph}
\Description{Temporal Graph}
\label{fig:pathexamplelong}
\end{figure}

% \textcolor{blue}{In Figure~\ref{fig:pathexample}, the sequence of (3,5), (8,3) edges is an earliest arrival path from $v_1$ to $v_5$, with respect to the given ready time $3$. Further, the sequence of (5,4), (9,3) edges is a fastest  path from $v_1$ to $v_5$. In the \textsc{eat} and \textsc{fpd} problems, the goal is to compute the earliest arrival times and fastest path durations from a source vertex to all the vertices.} 

% \begin{figure}[H]
% \centering
% \includegraphics[scale=0.2]{Images/temporalGraph.pdf}
% \caption{Temporal Graph}
% \label{fig:pathexample}
% \end{figure}

% \begin{figure}
%   \centering
%   % \includesvg{Images/temporal.svg}
%   \includegraphics{Images/Temporal_Graph.eps}
%   \caption{Temporal Graph}
%   \label{pathexampl}
% \end{figure}

\section{Useful Dominating Paths and ESDG } \label{Section_Graph_Transformation}

In this section, we first characterize the earliest arrival paths and fastest paths in temporal graphs. In other words, we define the notion of \emph{useful dominating paths} and show that every earliest arrival path and every fastest path is a useful dominating path. Further, we use the known transformation to convert a temporal graph to an equivalent directed acyclic graph(\textsc{dag})\cite{ESDG}. Later, we prove that all the earliest arrival paths and fastest paths are preserved in the transformed \textsc{dag}, which helps to design efficient algorithms. 

%A path in static graphs can be represented by a sequence of either vertices or edges. In temporal graphs, a sequence of vertices and a sequence of connections need not represent the same, due to the existence of multiple connections between two vertices. 
Now, we describe about useful dominating paths. 
A sequence $(v_1, v_2, v_3, \ldots, v_k)$ of vertices in $G$ is a \emph{route}, if there is an edge between every two consecutive vertices in the sequence.
For a path $P = (e_1, e_2, \ldots, e_k)$, and a route $r=(v_1, v_2, v_3, \ldots, v_{k+1})$,  we say that $P$ \emph{goes through} $r$, if for each $i$, $1\leq i\leq k$, left and right end vertices of $e_i$ are $v_i$ and $v_{i+1}$, respectively.
%$\leftvertex(e_i) = v_i$, and $\rightvertex(e_k) = v_{k+1}$.
A sub-path of a path is a \emph{prefix path} if both the paths start at the same vertex. Let $\mathbb{P}(s,z,r,t)$ denote the set of paths that depart at time $t$ from  $s$,  go through the route $r=(s=u_1, \ldots, u_k=z)$ of vertices and reach $z$.
A path $P$ in $\mathbb{P}(s,z,r,t)$ is a \emph{dominating path} if, for every path $Q$ in $\mathbb{P}(s,z,r,t)$, 
%from $s$ to $z$ that departs at time $t$ on route $r$, and 
$\arr(P) \leq \arr(Q)$.
% the arrival time at the destination for $P$ is greater than or equal to the arrival time for $Q$.
A dominating path $P$ in $\mathbb{P}(s,z,r,t)$  is a \emph{useful dominating path} if 
for every prefix path $P'$ of $P$, $P'$ is a dominating path.

In the route $r = (v_1, v_2, v_5, v_8, v_9)$ illustrated in Figure~\ref{fig:pathexamplelong}, the sequence of edges $(3,2),(6,3),(9,2),(11,4)$ constitutes a path denoted as $P$. Additionally, the edges $(3,3),(6,3),(9,2),(11,4)$ form a distinct path referred to as $Q$, and the edges $(3,4),(7,3),(10,3),(16,3)$ form yet another path called $R$.
 Since, $\arr(R) > \arr(P)$, $R$ does not meet the criteria of a dominating path. 
Upon considering the prefix paths $P'$ and $Q'$, of $P$ and $Q$, respectively on route $r = (v_1,v_2)$, it becomes apparent that $\arr(Q') > \arr(P')$. Consequently, it is established that $P$ is a useful dominating path, whereas $Q$ is dominating, but not a useful dominating path.
 
% \textcolor{blue}{In Figure~\ref{fig:usefulDominatingPath}, the sequence of $(e_1, e_2, \ldots, e_k)$ edges forms a useful dominating path $P$ on the route $r = (v_1, v_2, v_3, \ldots, v_{k+1})$. In both cases, $Q$ is not a useful dominating path. In case 1, path $Q$ is not a dominating path because $\arr(P) < \arr(Q)$. In case 2, for the prefix paths $P'$ and $Q'$ on route $r = (v_1, v_2, v_3)$, $\arr(P') < \arr(Q')$.}

% \begin{figure}[H]
% \centering
% \includegraphics[width=\columnwidth]{Images/usefulDominatingPath.pdf}
% \caption{useful Dominating Path}
% \label{fig:usefulDominatingPath}
% \end{figure}

%from $s$ to $z$, departing at time $t$ on route $r$ is a \emph{useful dominating path} if, for every prefix path $Q$ of $P$ is also a dominating path.

% \noindent \subsection{Notation and Terminology:}
%  % - route is a sequence of vertices
%  % - journey is a sequence of connections
%  % - journey on route: going through the route
%  % - dominating journey on a route
%  % - useful dominating journey on a route

\noindent
\begin{lemma}
\label{lemmaUsefulDominatingPath}
If there exists a path from $s$ to $z$ that departs at time $t$ on a route $r$, then there exists a  useful dominating path from $s$ to $z$ that starts at time $t$ on the route $r$.
%Let P be a path from s to z. Then there exists a  useful dominating path Q from $s$ to $z$ such that route(Q)=route(P), dep(Q)=dep(P) and arr(Q) \leq arr(P).
%Let $r$ be a route from $x_1$ to $x_k$. If $P_{x_1,x_k,r}(t) \neq \emptyset$, then there exists a useful dominating path from vertex $x_1$ to $x_k$ at time $t$ on route $r$.
\end{lemma}

% \noindent
\begin{proof}
Let  $P$ be a dominating path that departs at $t$, on route $r=(s=u_1, \ldots, u_k=z)$. We claim that there exists a useful dominating path on $r$. 
In the base case $k=1$, and thus the claim holds true, because the number of prefix paths of $P$ is one.
By the induction hypothesis, the claim holds true for $k-1$. In other words, if there is a path from $u_1$ to $u_{k-1}$ that starts at time $t$ on route $r$, then there exists a useful dominating path $Q$ from $u_1$ to $u_{k-1}$ that departs at $t$ on route $r$. Let $Q'$ be the sub path of the dominating path $P$ from $u_1$ to $u_{k-1}$. 
Now we shall replace $Q'$ of $P$ with $Q$ and the resultant dominating path is useful.
\end{proof}

If all the prefix paths of a dominating path on a route are dominating, then a prefix path of any prefix path is dominating. This results the following corollary. 
\begin{corollary}
\label{corollaryPrefixOfUsefulDominatingPath}
Every prefix path of a useful dominating path is a useful dominating path.
\end{corollary}

An  earliest arrival path is a dominating path on a route. Similarly, a fastest path is a dominating path on a route. Also, the departure times (and arrival times) of both dominating and useful dominating paths on the same route are equal. From these observations along with Lemma~\ref{lemmaUsefulDominatingPath},  we have the following corollaries. These corollaries helps to design a pruning strategy, by which exploring the useful dominating paths in temporal graphs will be sufficient to find earliest arrival and fastest paths.

\begin{corollary}
\label{corollaryEPisUseful}
For a temporal graph $G$, a ready time $rt$, and a source vertex $s$, let $P$ be an earliest arrival path from $s$ to $z$ such that $dep(P) \geq rt$.
% that starts at time t.
Then there exists a  useful dominating path $Q$ from $s$ to $z$ such that $route(Q)=route(P)$, $\dep(Q)=\dep(P)$ and $\arr(Q)=arr(P)$.
\end{corollary}

\begin{corollary}
\label{corollaryFPisUseful}
Let $P$ be a fastest path from $s$ to $z$ in a temporal graph $G$. Then there exists a  useful dominating path $Q$ from $s$ to $z$ such that $route(Q)=route(P)$, $\dep(Q)=\dep(P)$ and $\arr(Q)=\arr(P)$.
%Let P be an earliest arrival path from s to z. Then there exists a  useful dominating path Q from $s$ to $z$ such that route(Q)=route(P), dep(Q)=dep(P) and arr(Q)=arr(P).
\end{corollary}

%\begin{lemma} \label{lemmaFPisUseful} Let $P$ and $Q$ be dominating and useful dominating paths, respectively from $s$ to $z$ that start at a departute time $t$ on a route $r$. Then, $\arr(P) = \arr(Q)$. \end{lemma} \begin{proof} Every useful dominating path on $r$ is a dominating path on $r$, and thus we have $\journey(P) = \journey(Q)$. Similarly, a fastest path on $r$ is  a dominating path on $r$, and hence $\journey(P) = \journey(R)$. \end{proof}

%For a temporal graph, the ESDG graph $\Tilde{G}$ is defined as follows: 

Moving forward, we describe a graph transformation to preserve dominating paths. A temporal graph $G$ can be converted to an equivalent edge scan dependency graph (\textsc{esdg}) $\Tilde{G}$, in which all the edges of $G$ are treated as vertices or \emph{nodes} in $\Tilde{G}$, and  edges or \emph{dependencies} are added between  the nodes based on certain constraints. Ni et al. have developed this transformation to solve \textsc{eat} in the parallel setting \cite{ESDG}. In this work, we propose one to one mapping between useful dominating paths in $G$ to paths in $\Tilde{G}$ in the following lemma, and further use this characterization to design efficient algorithms.

Now, we shall look at the transformation from a temporal graph to the corresponding \textsc{esdg}. For a temporal graph $G = (V(G), E(G))$, the  node set $V(\Tilde{G})$, and the edge set
$E(\Tilde{G})$ are depended as follows: $V(\Tilde{G}) = \{ v_{e} \mid e \in E(G) \}$, 
$E(\Tilde{G}) = \{ (v_{e}, v_{f}) \mid e=(u,v,\alpha, \omega), f=(v,w,\alpha', \omega') \in E(G)$, and  no edge $(v,w,\alpha'', \omega'')$ exists such that $\alpha''\geq \omega$, and $\omega''<\omega'$ $\}$.
For each edge $e =(u,v,t,\lambda)$ in $G$, 
there is a node $v_e$ in $\Tilde{G}$, and we define four values namely left vertex, right vertex, departure time and arrival time of node $v_e$ as follows: 
$\leftvertex(v_e) = u$, $\rightvertex(v_e)=v$, $\dep(v_e)=t$, $\arr(v_e)=t+\lambda$. For the graph shown in Figure~\ref{fig:pathexamplelong}, the corresponding edge-scan-dependency graph is illustrated in Figure~\ref{Fig:ESDG}, in which each node of $\Tilde{G}$ is associated with a departure time and an arrival time.

\begin{figure}[H]
  \centering
  \includegraphics[width=\columnwidth]{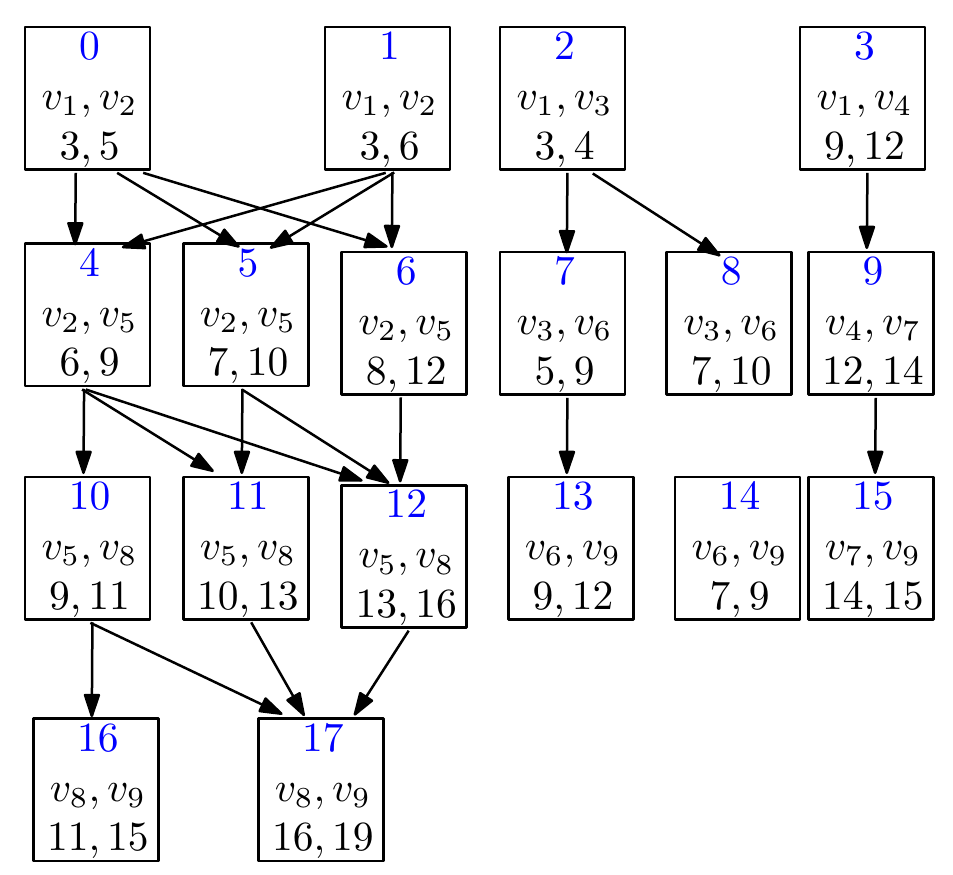}
  \caption{An edge-scan-dependency graph}
  \label{Fig:ESDG}
  \Description{Temporal Graph}
\end{figure}

\begin{lemma} 
\label{lemmaESDGTransformation1}
%Let $r=(v_1, v_2, \ldots, v_k)$ be a route in $G$. Let $P$ be a journey on $r$. Let $e$ and $f$ are the first and last connections of $P$. If $P$ is useful dominating, then there is a path from $v_e$ to $v_f$ in $\Tilde{G}$. If $P$ is not useful dominating, then there is no path from $v_e$ to $v_f$ in $\Tilde{G}$. \textcolor{red}{refine later} \noindent There is a useful dominating path that connects $e$ and $f$ on a route in $G$ if and only if there is a path from $v_e$ to $v_f$ in $\Tilde{G}$.
A sequence $e_1, \ldots, e_k$ of $k$ edges  in $G$ is a useful dominating path  if and only if a sequence $v_{e_1}, \ldots, v_{e_k}$ of $k$ vertices  in $\Tilde{G}$ is a path.
Further, the journey time of the path from $e_1$ to $e_{{k-1}}$ in $G$ is equal to the journey time of the path from $v_{e_1}$ to $v_{e_{k-1}}$ in $\Tilde{G}$.

\end{lemma}
\begin{proof}
    Let $e_{k-1}=(u, v, \alpha, \omega)$ and $e_k=(v, w, \alpha', \omega')$, and we use this notation in the following proofs.
We first prove the forwarding direction of this lemma using induction on $k$. %In the base case, $k$ equals to 1. From the construction of an ESDG graph, every node in $\Tilde{G}$ correspond to a connection in $G$ and vice versa. Thus the base case holds true. 
In the base case,  $k$ equals to 1, and  $e_1$ is a useful dominating path. Because, each edge in $G$ is represented as a node in ESDG graph $\Tilde{G}$, there is a path of length 0 from $v_e$ to $v_e$ in  $\Tilde{G}$. Thus the base case holds true. 
    %Consider a useful dominating path $P$ in $G$ on $r$, which is going through $k$ vertices. 
    Let $P$ be the path formed with the sequence $e_1, \ldots, e_k$ of $k$ edges.
    Let $P'$ be the subpath of $P$ on the first $k-1$ edges.
    $P'$ is useful dominating path due to Corollary~\ref{corollaryPrefixOfUsefulDominatingPath}.
    %the definition of useful dominating path. 
    %The path $P'$  appears in $\Tilde{G}$, 
    %Let $f'=(v_{k-2}, v_{k-1}, \alpha', \omega')$ and $f=(v_{k-1}, v_{k}, \alpha, \omega)$ be the last two connections of $P$. 
    %Let $e_{k-1}=(u, v, \alpha, \omega)$ and $e_k=(v, w, \alpha', \omega')$ be the last two connections of $P$. 
    Due to induction hypothesis, there is a path $Q$ from $v_{e_1}$ to $v_{e_{k-1}}$ in $\Tilde{G}$.
    Because $P$ is a dominating path, there does not exist any edge $(v, w, \alpha'', \omega'')$, st. $\alpha'' \geq \omega$ and $\omega'' < \omega'$. 
    Hence, we would have added an edge between $e_{k-1}$ and $e_k$ in $\Tilde{G}$. 
    Thus, the concatenation of path $Q$ from $v_{e_1}$ to $v_{e_{k-1}}$ and edge $(v_{e_{k-1}},v_{e_k})$, results a path from $v_{e_1}$ to $v_{e_{k}}$ in $\Tilde{G}$. 

   Now, we prove the backward direction of the lemma, using induction on $k$.
   The base case holds true when $k=1$, because there is a one to one mapping between edges in $G$ and nodes in $\Tilde{G}$.
    Let $P$ be the path formed with the sequence $v_{e_1}, \ldots, v_{e_k}$ of $k$ nodes.
    Let $P'$ be the subpath of $P$ on the first $k-1$ edges.
    Due to induction hypothesis, there is a useful dominating path $Q$ from $e_1$ to $e_{k-1}$ in $G$.
Since  an edge is added from $e_{k-1}$ to $e_k$ in $\Tilde{G}$, there does not exist any edge $(v, w, \alpha'', \omega'')$,
st. $\alpha'' \geq \omega$ and $\omega'' < \omega'$. 
Therefore, $e_1, \ldots, e_{k-1}, e_{k}$ is a useful dominating path in $G$. 

The departure times of  $v(e_1)$ and $e_1$ are equal, and the arrival times $v(e_k)$  and $e_k$ are equal, from the transformation. Thus the  journey times of both paths provided in the second part of the lemma are same.
\end{proof}

In this section, we first provide the pseudo-code to solve \textsc{earliest arrival time} problem  in Algorithm~\ref{algoEarliestPathVersion1}, and prove the correctness in Theorem~\ref{TheoremEAT}. Later, the pseudo-code of Algorithm~\ref{algoEarliestPathVersion1} is enhanced in Algorithm~\ref{algoEarliestPathVersion2}, to bound the running time.
% using necessary data-structure in Algorithm~\ref{algoEarliestPathVersion2}.
%for calculating the earliest arrival time from a source vertex $s$ to all other vertices in $G$, for a given ready time $rt$.

% \textcolor{red}{Multi source BFS; simple and elegant algorithm}
The key idea in Algorithm~\ref{algoEarliestPathVersion1} is to explore from those nodes in $\Tilde{G}$ that correspond to edges in $G$, such that their left end vertex is the source vertex, and the departure time is at least $rt$. During the exploration, we identify all the reachable nodes and update the arrival times of their right end vertices, if the new arrival time is better than the existing one.

\begin{theorem} \label{TheoremEAT}
    Given an ESDG $\Tilde{G}$ of a temporal graph $G$, a source vertex $s$ in $G$ and a ready time $rt$,  Algorithm \ref{algoEarliestPathVersion1} correctly computes the earliest arrival time from $s$ to every vertex in $G$.
\end{theorem}

\begin{proof}
    Let $z$ be a vertex in $G$ such that $z\neq s$. We now prove that, after all the iterations of Algorithm~\ref{algoEarliestPathVersion1} are over, $\eat[z]$ is equal to the earliest arrival time from $s$ to $z$. 
    In Algorithm~\ref{algoEarliestPathVersion1}, for every node $x$ in  $\Tilde{G}$ such that $\leftvertex(x)=s$ and $\dep(x)\geq rt$, we traverse from $x$, go through all the paths to find all reachable nodes in $\Tilde{G}$.
    The arrival times of all these paths are considered and stored the minimum one in Line~\ref{algoLineEATUpdatev1}.
    Also, these paths in $\Tilde{G}$ precisely correspond to those useful dominating paths from $s$ to $z$ in $G$ whose departure time is at least $rt$, due to Lemma~\ref{lemmaESDGTransformation1}.
    From Corollary~\ref{corollaryEPisUseful}, it is sufficient to consider useful dominating paths from $s$ to $z$ that depart at time at least $rt$, to obtain an earliest arrival path from $s$ to $z$,  for the given departure time $rt$. Thus the theorem holds true.  
\end{proof}

\begin{algorithm2e}
\caption{Earliest Arrival Time in Temporal Graph}
\label{algoEarliestPathVersion1}
\SetAlgoLined
\KwIn{A source vertex $s$, ready time at source $rt$ and edge scan dependency graph $\Tilde{G}$ of a temporal graph $G$.}
\KwOut{ For each vertex $z$ in $G$, the earliest arrival time from $s$ to $z$.}

\lFor{each vertex $z$ in $V(G)-s$}{ $eat[z]$ = $\infty$ } $eat[s] = rt$ \;
\For{each node $x$ in  $\Tilde{G}$ such that $\leftvertex(x)=s$ and $\dep(x) \geq rt$ }
{
    \For{each path \textsc{p} from $x$ to a node $y$ in $\Tilde{G}$}
    {
          $\eat[\rightvertex(y)] = \min\{ \eat[\rightvertex(y)], \arr(y) \}$ \label{algoLineEATUpdatev1} \;
    }

}
\end{algorithm2e}

The correctness of Algorithm~\ref{algoEarliestPathVersion1} to solve  \textsc{earliest arrival time} problem follows from Theorem~\ref{TheoremEAT}.
Moving forward, Algorithm~\ref{algoEarliestPathVersion1} is enhanced to Algorithm~\ref{algoEarliestPathVersion2} by processing the vertices and edges of $\Tilde{G}$ at most once, across multiple breadth first search traversals to bound the running time. 
This does not disturb the correctness of the algorithm, because processing a vertex $x$ in $\Tilde{G}$ multiple times does not change the arrival time of $\rightvertex(x)$.
%When a vertex in $\Tilde{G}$ is visited for the first time, their out going neighbours are inserted .
%\textcolor{red}{are we pruning here?} \textcolor{blue}{connecting sentence is required here accordingly}

\begin{algorithm2e}
\caption{Earliest Arrival Time in Temporal Graph}
\label{algoEarliestPathVersion2}
\SetAlgoLined
\KwIn{A source vertex $s$, ready time at source $rt$ and edge scan dependency graph $\Tilde{G}$ of a temporal graph $G$.}
\KwOut{ For each vertex $z$ in $G$, the earliest arrival time from $s$ to $z$.}

\lFor{each vertex $z$ in $V(G)-s$}{ $eat[z]$ = $\infty$ } $eat[s] = rt$ \; \label{LineinitEAT}
\lFor{each node $x=(u,v,t, \lambda)$ in $V(\Tilde{G})$}{ $visited[x]=false$ }                  \label{LineinitVisited}
\For{each node $x=(u,v,t, \lambda)$ in $V(\Tilde{G})$ such that $\leftvertex(x)=s$ and $\dep(x) \geq rt$ } 
{ \label{algoLineEATsrcnodeprocessing}
    $q.insert(x)$ ;   $visited[x] = true$  \;
    \While{$(|q|\geq 1)$}
    {
        $x = q.pop()$ \;
        $eat[\rightvertex(x)] = min(eat[\rightvertex(x)], \arr(x))$ \;   \label{algoLineEATUpdate}
        \For{\textbf{each} neighbor $y$ of $x$ in $V(\Tilde{G})$ such that $visited[y]=false$}{ $q.insert(y)$; $visited[y] = true$  } 
    }     
}
\end{algorithm2e}

Algorithm~\ref{algoEarliestPathVersion2} works as follows. Given a source vertex $s$, ready time $rt$ at the source, and an \textsc{esdg} graph $\Tilde{G}$, the algorithm initializes the earliest arrival time $\eat[z]$ for each vertex $z$ in $V(G)$. The earliest arrival time represents the minimum arrival time to reach vertex $z$ from the source vertex $s$. During the initialization phase in Line~\ref{LineinitEAT}, all vertices, except the source, are set to have an earliest arrival time of $\infty$, while the source vertex is set to $rt$. 
Also, $visited[x]$ for every node $x$ in $\Tilde{G}$ is set to \emph{false} in Line~\ref{LineinitVisited}, to indicate that all nodes in the beginning are not visited. We now, go through all the nodes $x$ in $\Tilde{G}$ such that $\leftvertex(x)=s$ and $\dep(x) \geq rt$, and perform a breadth first kind of traversal from $x$, as follows.
The queue $Q$ is also initialized with  $x$ and $visited[x]$ is updated with true.
%where $u=s$ and the departure time at the source vertex is greater than or equal to $t$.
After the initialization, the algorithm proceeds to process edges by deleting them one by one from the queue and perform relaxation in Line~\ref{algoLineEATUpdate}. 
Further, all the unvisited neighbours of $x$ are added to the queue. These two steps are repeated until the queue becomes empty. Finally, for every vertex $z$ in $G$, $\eat[z]$ holds the earliest arrival time to reach $z$.

\section{Efficient Algorithm for Fastest Path duration}
In this section, we propose efficient  algorithms for computing the fastest path duration from the given source vertex to all the  vertices at high-level in Algorithm~\ref{algoFastestPathVersion1} and the complete details in Algorithm~\ref{algoFastestPathVersion2}, and the algorithm correctness is proved in Theorem~\ref{theoremFPDAlgo}.

% \textcolor{blue}{Efficiency of many graph algorithms in practice is based on effective way of handling pruning and redundant computations. The following table shows four possibile ways of designing algorithms. }
% \textcolor{red}{Multiple algorithms can be desinged based on ESDG.}

% \textcolor{blue}{Explain 4 ideas crisply.}

% \textcolor{blue}{We first focus to handle pruning.}

The fundamental idea in Algorithm~\ref{algoFastestPathVersion1} is to perform a graph traversal from those nodes in $\Tilde{G}$ that correspond to the outgoing edges of a source vertex in $G$. During each graph traversal phase, we propagate the starting time to all the reachable nodes and update the journey times of their right end vertices, if the new journey time is better than the existing one.

\begin{algorithm2e}
\caption{1  - All Fastest Path Algorithm using ESDG}
\label{algoFastestPathVersion1}
\SetAlgoLined
\KwIn{A source vertex $s$, and an edge scan dependency graph $\Tilde{G}$ of a temporal graph $G$}
\KwOut{For each vertex $z$ in $G$, the  fastest duration from $s$ to $z$.}

\lFor{each vertex $z$ in ${G}$ $\setminus$ $s$}{ $\journey[z] = \infty$} \label{startinitV1}
$\journey[s]=0$ \;
\For{each node $x$ in $\Tilde{G}$, such that $\leftvertex(x)=s$} 
{ \label{LineAlgoV1Phasestart}
    \For{each each path \textsc{p} from $x$ to a node $y$ in $\Tilde{G}$}
    {
          $\journey[\rightvertex(y)] = \min\{ \journey[\rightvertex(y)], \arr(y)-\dep(x) \}$ \label{algoLineJourneyUpdate} \;
    } 
} \label{LineAlgoV1Phaseend}
\end{algorithm2e}

\begin{theorem} \label{theoremFPDAlgo}
     Given an ESDG $\Tilde{G}$ of a temporal graph $G$, a source vertex $s$ in $G$, Algorithm \ref{algoFastestPathVersion1} correctly computes the fastest path duration from $s$ to every vertex $z$ in $G$. 
\end{theorem}

\begin{proof}
Let $z$ be a vertex in $G$ such that $z\neq s$. We now show that, after all the iterations of Algorithm~\ref{algoFastestPathVersion1} are over, $\journey[z]$ is equal to the journey time of a fastest path from $s$ to $z$. 
In Algorithm~\ref{algoFastestPathVersion1}, for every node $x$ in  $\Tilde{G}$ such that $\leftvertex(x)=s$, we traverse from $x$, go through all the paths to find all  reachable nodes in $\Tilde{G}$.
%The key step in the algorithm is to traverse through all the paths that originate at x and find all the all reachable nodes in $\Tilde{G}$
%find all paths from $x$ to the all reachable nodes in $\Tilde{G}$, where $x \in \delta_{out}(s)$. 
Out of all these paths, let us observe all those paths that end at $y$, such that  $\rightvertex(y)=z$.
%  $y\in \delta_{in}(z)$.
The journey times of all these paths are considered and stored the minimum one in $\journey[z]$ in Line~\ref{algoLineJourneyUpdate}.
Also, these paths in $\Tilde{G}$ precisely correspond to the useful dominating paths from $s$ to $z$ in $G$ due to Lemma~\ref{lemmaESDGTransformation1}.
%Since every fastest path is a useful dominating path, it is sufficient to consider all useful dominating paths.
%Processing all useful dominating paths from $s$ to $z$ in $G$ is sufficient to obtain a fastest path from s t z, because of Leamma~\ref{}.
%For every fastest path p from s to z, there exist a useful dominating path q from s to z, such that journey(p) = journey(q)
From Corollary~\ref{corollaryFPisUseful}, it is sufficient to consider useful dominating paths from $s$ to $z$ to obtain a fastest path from $s$ to $z$.
Thus the theorem holds true.
\end{proof}

The correctness of Algorithm \ref{algoFastestPathVersion1} is formally established in Theorem \ref{theoremFPDAlgo}.

% \textcolor{blue}{We now focus to handle redundant nodes.}

Turning our attention to enhance the efficiency, we avoid processing the same node multiple times. 
%further optimizations to Algorithm \ref{algoFastestPathVersion1}, enhancing its efficiency, particularly when nodes are processed in multiple paths. 
For instance, consider two nodes, $x_i$ and $x_j$, in $\tilde{G}$ with $\leftvertex(x_i) = \leftvertex(x_j) = s$, selected during iterations $i$ and $j$ of the outer loop in Line~\ref{LineAlgoV1Phasestart} of Algorithm~\ref{algoFastestPathVersion1}. Let $y$ be a reachable node from both $x_i$ and $x_j$, with $P_i$ denoting the path from $x_i$ to $y$ and $P_j$ representing the path from $x_j$ to $y$. If $\dep(x_j) \leq \dep(x_i)$, then it follows that $\journey(P_j) \geq \journey(P_i)$ since $\arr(P_i) = \arr(P_j)$. This insightful observation leads to a critical optimization: a node processed in the $i^{th}$ iteration of the outer loop need not be processed again in the $j^{th}$ iteration if $\dep(x_j) \leq \dep(x_i)$, as the journey from $s$ to $\rightvertex(y)$ does not decrease. To harness the advantages of this optimization, we traverse nodes in the edge scan dependency graph whose left vertex is the source vertex, by ordering them in the non-increasing order based on their departure times. This avoids the redundant computation of processing the same nodes multiple times.

 %   Let $x$ be a node in $\tilde{G}$ that was visited in the $j^{th}$ iteration, where $1 \leq j \leq i-1$.  Let $P'$ be the corresponding path. Furthermore, let $P$ be a path from $s$ to $t_i$ via $x$, with $x$ ending at $v_i$. We can observe that $\journey(P) \geq \journey[v_i]$. This is because both $P$ and $P'$ have the same arrival time, denoted by $arrivalTime(P') = arrivalTime(P)$. Additionally, we have $departureTime(P) < departureTime(P')$ and $\journey(P') \geq \journey[v_i]$. Therefore, extending the traversal by considering $x$ and the path $P$ does not decrease the journey time required to reach $v_i$. As a result, we can early terminate and prune the visited nodes in the subsequent iterations, knowing that the current path $P'$ already provides a faster journey time to reach $v_i$.

\begin{algorithm2e}
\caption{1  - All Fastest Path Algorithm using ESDG}
\label{algoFastestPathVersion2}
  \SetAlgoLined
\KwIn{A source vertex $s$, and an edge scan dependency graph $\Tilde{G}$ of a temporal graph $G$}
\KwOut{For each vertex $z$ in $G$, the  fastest duration from $s$ to $z$.}

\lFor{each vertex $z$ in ${G}$ $\setminus$ $s$}{ $\journey[z] = \infty$} \label{startinitFPD}
$\journey[s]=0$ \;
\lFor{each node $x$ in $\Tilde{G}$}{ $st[x] = - 1 $ }  \label{endinitFPD}
\For{each node $x$ in $\Tilde{G}$ such that $left(x)=s$ and decreasing order of $x.t$} 
{ \label{algoLineFPDsrcnodeprocessing}
    $q.insert(x)$ \label{LineAlgoInit_v1};
    $st[x] = x.t$ \;
    \While{$(|q|\geq 1)$}
    {
        $x = q.pop()$ \;
        $\journey[x.v] = min(\journey[x.v],\arr(x) - st[x])$ \;
        \lFor{each neighbor $y$ of $x$  in $\Tilde{G}$ such that $st[y] == -1 $} { $st[y] = st[x]$; $q.insert(y)$ }
    }
}
\end{algorithm2e}

%In this algorithm \ref{algoFastestPathVersion2}, we utilize a single queue to process connections, eliminating the need for maintaining a list of tuples at each vertex. By employing a pruning technique, we process only a minimal number of connections.
We now provide the detailed description of our algorithm to solve \textsc{fastest path duration}, whose pseudo-code is given in Algorithm~\ref{algoFastestPathVersion2}.
For each vertex $z$ in $G$, the variable $\journey[z]$ stores the journey time, representing the fastest path duration from the source vertex $s$ to $z$.
For each node $x$ in $\Tilde{G}$, we use $st[x]$ to store the starting time of a journey that departs as late as possible from $s$ and  reaches node $x$.
We use $q$ to denote the queue data structure, which helps to perform breadth first kind of traversal.
During the initialization phase (lines~\ref{startinitFPD}-\ref{endinitFPD}), we set the fastest path duration of all vertices, except the source vertex, as $\infty$. The fastest path duration of the source vertex is set to $0$.
Also, the starting time for all the nodes of $\Tilde{G}$ is set to -1, to indicate that none of them are visited in the beginning. 
We now go through each node $x$  in $\Tilde{G}$ such that $x.u=s$ in decreasing order based on the departure time of the nodes, and perform breadth first search kind of traversal to identify reachable nodes from $x$ as follows.
We insert $x$ in the queue $q$, and initialize $st[x]$ with its departure time. 
As long as the queue is not empty, we extract a node $x$ from the queue and update  $\journey[x.v]$ if $x.t+x.\lambda-st[x])$ is lesser than $\journey[x.v]$. Afterwards, we insert each neighbor $y$ of $x$ into the queue $q$ and update $st[y]$ with $st[x]$, if $y$ has not been visited yet.
%with the maximum value between $st[y]$ and $st[x]$. This algorithm only processes the nodes (connections) that are inserted into the queue. 
%\textcolor{blue}{By checking if the respective connection starting time is not set (i.e., -1), we optimize the algorithm to reduce unnecessary connection processing, thereby lowering the processing time.} 
Finally, $\journey[z]$ holds the fastest path duration from $s$ to any vertex $z$ due to Theorem~\ref{theoremFPDAlgo}.

%\textcolor{blue}{run time analysis:} The time complexity of the \textsc{EAT} and \textsc{FPD} algorithm is $O(m + k*m)$, where $m$ represents the number of edges in the temporal graph, and $k$ is the average temporal out-degree of a vertex. Given that $k$ is a constant with a maximum value of 3, we can simplify the run time analysis to $O(m)$, making it linear in terms of the number of edges in the temporal graph.
\noindent
\textbf{Time Complexity Analysis.} We perform breadth first search kind of traversals from multiple vertices of $\Tilde{G}$, in Algorithm~\ref{algoEarliestPathVersion1} and Algorithm~\ref{algoEarliestPathVersion2}. Although multiple breadth first search traversals are performed, we make sure that each vertex of $\Tilde{G}$ is inserted in the underlying queue and process their incident edges at most once, with the help of visited[ ] array and st[ ] array, in the respective algorithms. Thus the asymptotic running times of our algorithms is $O(|V(\Tilde{G})| + |E(\Tilde{G})|)$. From the \textsc{esd} graph construction, $|V(\Tilde{G})| = |E(G)|$. 
For each temporal edge $e=(u,w,t,\lambda)$ of $G$, the out-degree of $v_e$ in $\Tilde{G}$ is equal to the out-degree of $w$ in the static road network associated with $G$.
Also, the average out degree of a vertex in $\Tilde{G}$ is equal to the average out degree of the static road network associated with a public transport network $G$, which is denoted by $\Delta$. 
Consequently, the asymptotic running time of our algorithms is $O(m \times \Delta +n)$.
This value of $\Delta$ turned out to be a small constant in all the real world public transport networks. From Table~\ref{tab:AvgOutDegree}, we can observe that the max degree and average degree of vertices in real world transport networks are $61$ and $3$, respectively. Eventually, the time complexity of our algorithms is bounded by $O(m+n)$.

% \begin{table}[H]
% \centering
% % \resizebox{\columnwidth}{!}{
% \begin{tabular}{l|cc|}
% \cline{2-3}
%                                            & \multicolumn{2}{c|}{\textbf{Max Out Degree}} \\ \hline
% \multicolumn{1}{|c|}{\textbf{Data Sets}} & \multicolumn{1}{c|}{\textbf{Static Graph}} & \textbf{ESD Graph} \\ \hline
% \multicolumn{1}{|l|}{\textbf{Chicago}}     & \multicolumn{1}{c|}{17}         & 17         \\ \hline
% \multicolumn{1}{|l|}{\textbf{London}}      & \multicolumn{1}{c|}{7}          & 7          \\ \hline
% \multicolumn{1}{|l|}{\textbf{Los Angels}}  & \multicolumn{1}{c|}{7}          & 7          \\ \hline
% \multicolumn{1}{|l|}{\textbf{Madrid}}      & \multicolumn{1}{c|}{8}          & 8          \\ \hline
% \multicolumn{1}{|l|}{\textbf{New York}}    & \multicolumn{1}{c|}{3}          & 3          \\ \hline
% \multicolumn{1}{|l|}{\textbf{Paris}}       & \multicolumn{1}{c|}{61}         & 60         \\ \hline
% \multicolumn{1}{|l|}{\textbf{Peters burg}} & \multicolumn{1}{c|}{22}         & 22         \\ \hline
% \multicolumn{1}{|l|}{\textbf{Sweden}}      & \multicolumn{1}{c|}{43}         & 43         \\ \hline
% \multicolumn{1}{|l|}{\textbf{Switzerland}} & \multicolumn{1}{c|}{49}         & 49         \\ \hline
% \end{tabular}%
% % }
% \caption{Maximum Out-Degree Comparison between Static and ESD Graphs}
% \label{tab:max out degree}
% \end{table}
\section{Implementation Details}

% \textcolor{blue}{In this section, we discuss the \textsc{esd}-graph data structure and various optimizations followed while implementing Algorithms~\ref{algoEarliestPathVersion2} and~\ref{algoFastestPathVersion2}.}

In this section, we discuss the implementation details of the \textsc{esd}-graph data structure, and present various optimizations in the context of implementing Algorithm~\ref{algoEarliestPathVersion2} and Algorithm~\ref{algoFastestPathVersion2}.

% We have taken real-time public transportation data, usually available in the General Transit Feed Specification (GTFS) format, and transformed it into a Compressed Sparse Row (CSR) format. CSR ensures fast access to neighbour information, a critical need for our algorithms, enhancing efficiency across various applications. It utilizes two arrays: the offset array for quick neighbour access and the edge array for efficient neighbour node storage. This format excels in rapid neighbour information retrieval and graph operations, ideal for managing public transportation network data, as depicted in Figure 1.

%\noindent \textbf{Graph Representation :} We have taken real-time public transportation data, usually available in the General Transit Feed Specification \textsc{(GTFS)} format, and transformed it into an Edge Scan Dependency graph, as described in section 1. Finally, the transformed Edge Scan Dependency graph is stored in Compressed Sparse Row \textsc{(CSR)} format. \textsc{(CSR)} ensures fast access to neighbour information, a critical need for our algorithms. It utilizes two arrays: the offset array for quick neighbour access and the edge array for efficient neighbour node storage. This format excels in rapid neighbour information retrieval and graph operations, ideal for managing public transportation network data, as depicted in Figure. \ref{fig:CSRESDG}. Notably, the \textsc{(CSR)} format is highly space-efficient and accessing neighbour information is achieved in constant time.

\noindent \textbf{\textsc{esd}-graph data structure.} We process real-time public transportation data available in the General Transit Feed Specification \textsc{(GTFS)} format, and transform to an \textsc{esd} graph, as described in Section~\ref{Section_Graph_Transformation}.
We construct \textsc{esd}-graph data structure from a given temporal graph using the pre-processing algorithm illustrated in \cite{ESDG}.
We store an \textsc{esd} graph $\Tilde{G}$ in  our data-structure (offset, neighbour, left, right, departure, duration) whose parts are described below. The topology of $\Tilde{G}$ is captured using Compressed Sparse Row \textsc{(csr)} format.
\textsc{csr} ensures fast access to neighbour information, a critical need for our algorithms using two arrays offset[ ] and neighbours[ ]. For each vertex $v$ in $\Tilde{G}$, the neighbours of $v$ are located from position $offset[v]$ to $offset[v+1]-1$ in neighbours[ ]. For each vertex $v_e$ in $\Tilde{G}$, we maintain the following four attributes: $\leftvertex[v_e$] denotes the left vertex of $e$, $\rightvertex[v_e]$ denotes the right vertex of $e$, $\dep[v_e]$ denotes the departure time of $e$, and $\arr[v_e]$ denotes the arrival time of $e$. This format excels in rapid neighbour information retrieval and graph operations, ideal for managing public transportation network data. Notably, this format is highly space-efficient and accessing neighbour information is achieved in constant time.
%Efficiently implemented Line~\ref{algoLineEATsrcnodeprocessing} of Algorithm ~\ref{algoEarliestPathVersion2} and Line~\ref{algoLineFPDsrcnodeprocessing} of Algorithm ~\ref{algoFastestPathVersion2} by maintaining the source vertex related nodes in decreasing order of departure times. 
For each vertex $v$ in $G$, we maintain a sequence of vertices in $\Tilde{G}$, which are correspond to the outgoing edges of $v$. This helps to retrieve the necessary vertices in $\Tilde{G}$ in Line~\ref{algoLineEATsrcnodeprocessing} of Algorithm ~\ref{algoEarliestPathVersion2} and Line~\ref{algoLineFPDsrcnodeprocessing} of Algorithm ~\ref{algoFastestPathVersion2}, efficiently.

\noindent
\textbf{Optimizations.} We describe various optimizations that improve the running time of our algorithms in practice. 
In Line~\ref{algoLineEATUpdate} of Algorithm~\ref{algoEarliestPathVersion2}, if the earliest arrival time of a vertex is not updated, then we can ignore exploring its neighbours. This optimization is based on the following observation. If $e$ and $e'$ corresponds to incoming edges of a vertex in $G$ such that their arrival times are same, then the outgoing neighbours of $v_e$ and $v_{e'}$ are same in $\Tilde{G}$.

%In the first optimization, if two nodes, denoted as $P_i$ and $P_j$, share the same earliest arrival time, we can avoid processing the node again, as there's no additional benefit. This optimization is based on the fact that if two paths have the same arrival time at their destination, there's no need to recompute them.

% The second optimization is based on a condition involving the node's departure time and duration. If the sum of the node's departure time and duration is greater than or equal to the earliest arrival time at its neighboring node, we can safely skip inserting that neighboring node into the queue. This is because it would have been inserted during previous iterations, and its information is still valid. These optimizations help reduce redundant computations and enhance the efficiency of our algorithm.

% We implement a bit optimization technique for the visited array to further optimize algorithm~\ref{algoEarliestPathVersion2}. By maintaining a one-bit to represent the visited status of vertices, we can significantly reduce memory usage and improve processing speed. This optimization involves setting or clearing individual bits to mark visited vertices. This technique is beneficial when dealing with many vertices and allows us to efficiently keep track of visited nodes without the memory overhead of a traditional Boolean array. 

We employ a bit optimization technique in Algorithm~\ref{algoEarliestPathVersion2}.
In particular, an array $B$ of $n$ bits are used, and utilize a single bit of $B$ rather than one byte, to represent the visited status of a vertex. This optimization involves resetting all the bits of $B$ at the beginning  to initialize all the vertices of $\Tilde{G}$ as unvisited. Setting an individual bit in $B$  helps to mark a vertex as visited. These two operations can be performed in constant time. This optimization helps to achieve a notable reduction in memory consumption and an improvement in processing speed.
%where each vertex status is denoted using one bit rather than one byte, 
%This enhancement aims to boost algorithm performance. 
%This optimization involves two operations setting and clearing individual bits to mark a vertex as visited, which takes constant time. By utilizing a single bit to represent the visited status of a vertex, we achieve a notable reduction in memory consumption and an improvement in processing speed. This optimization involves two operations setting and clearing individual bits to mark a vertex as visited, which takes constant time. %This technique proves advantageous, particularly in scenarios with numerous vertices, enabling us to efficiently track visited nodes without incurring the memory overhead associated with a traditional Boolean array.}

\section{Experiments}

In this section, we discuss the experimental setup and various experiments carried out on Algorithms~\ref{algoEarliestPathVersion2} and~\ref{algoFastestPathVersion2}, and highlighting the speedup achieved over state-of-the-art algorithms. 

\noindent
\textbf{\textit{Technical Specifications and Data sets. }} The experimentation is conducted on a machine equipped with an \textsc{intel xeon e5-2620} v4 \textsc{cpu}, operating at a frequency of 2.20 GHz, featuring 32 GB of primary memory and 512 MB cache memory. The compiler used is gcc version 5.4.0. We have used nine different public transport network data sets \cite{OpenMobilityData,HaryanGPU} for our experiments. The statistics for each data set are given in Table \ref{tab: Dateset Statistics}.

\begin{table}[ht]
\centering
\begin{tabular}{|l|c|c|c|c|}
\hline
\multicolumn{1}{|c|}{Data Sets} & $|V(G)|$ & $|E(G)|=|V(\Tilde{G})|$ & $|E(\Tilde{G})|$ \\ \hline
Chicago     & 240   & 98157    & 44907     \\ \hline
London      & 20843 & 14064967 & 12103649  \\ \hline
Los Angels  & 13975 & 1979340  & 2320947   \\ \hline
Madrid      & 4689  & 1994688  & 2753161   \\ \hline
New York     & 987   & 514390   & 499713    \\ \hline
Paris       & 411   & 1068284  & 50965       \\ \hline
Peters burg  & 7573  & 4437010  & 6038003  \\ \hline
Sweden      & 45727 & 6567745  & 12144520  \\ \hline
Switzerland & 29870 & 9261315  & 12147435  \\ \hline
\end{tabular}
\caption{Data Set Statistics}
\label{tab: Dateset Statistics}
\end{table}

\noindent
\subsection{Performance of the Earliest Arrival Time Algorithm}

%assess their computational time. 
%The first one is the time ordered sequence of edges based Earliest Arrival Time algorithm, abbreviated as Edge-Stream-EAT, which processes all edges one by one to calculate the earliest arrival time from a given source vertex to the other vertices in the graph. 
%The second is the TRG-based Earliest Arrival Time algorithm, abbreviated as TRG-EAT. This algorithm is obtained from the fastest path algorithm in \cite{Sahani2021} with minor changes.

We implemented the state-of-the-art algorithms and our algorithm to solve \textsc{eat} problem in C++.
Specifically, we examined two state of the art algorithms proposed in \cite{wu_2016_Fast_EAT_Temporal, Sahani2021}. 
The  first earliest time algorithm, abbreviated as \textsc{edge-stream-eat}, is based on an edge stream, in which the temporal edges are relaxed  in non-decreasing order, based on their departure time \cite{wu_2016_Fast_EAT_Temporal}. 
The second earliest arrival time algorithm is based on a time respecting graph abbreviated as \textsc{trg-eat}, 
is obtained from the fastest path algorithm in \cite{Sahani2021} with minor changes.

We run state-of-the-art and proposed algorithms on $100$ generated random queries, each consisting of two values: a source vertex and a ready time. The source vertices are randomly selected from $0$ to $n$, where $n$ denotes the number of vertices in the underlying graph, and the corresponding ready times are chosen randomly within the range of $0$ to $100$. We then use these generated queries to run all three algorithms on nine public transportation data sets, measuring the average query running time in milliseconds. Table \ref{tab:EATRunningTime} presents the average running times of our proposed algorithm. The speedups of our approach in comparison to the state-of-the-art algorithms edge-scan based algorithm ~\cite{wu_2016_Fast_EAT_Temporal} and \textsc{trg} based algorithm~\cite{Sahani2021} shown in the Fig. \ref{EATSpeedup}.

Computing the earliest arrival time using our approach, we achieved \textit{$~183 \times$ maximum and $~24 \times$ average speedup} over algorithm ~\cite{wu_2016_Fast_EAT_Temporal} and \textit{$~48 \times$ maximum and $~24 \times$ average speedup} over ~\cite{Sahani2021}.
% while Fig. \ref{EATSpeedup} illustrates the speedups of our approaches compared to the connection-scan algorithm~\cite{wu_2016_Fast_EAT_Temporal}.

% \begin{table}[H]
% \centering
% \resizebox{\columnwidth}{!}{
% \begin{tabular}{|l|c|c|c|}
% \hline
% \multicolumn{1}{|c|}{\textbf{Data Sets}} &
%   \textbf{\begin{tabular}[c]{@{}c@{}} Connection stream \end{tabular}} &
%   \textbf{\begin{tabular}[c]{@{}c@{}} ESDG \end{tabular}}&
%   \textbf{\begin{tabular}[c]{@{}c@{}}Our Approach \end{tabular}} \\ \hline
% \textbf{Chicago}     & 0.7064   & 0.1141   & 0.0528  \\ \hline
% \textbf{London}      & 100.6210 & 364.1140 & 25.6239 \\ \hline
% \textbf{Los Angels}  & 13.7513  & 49.7495  & 11.3584 \\ \hline
% \textbf{Madrid}      & 12.8811  & 44.7953  & 2.7957  \\ \hline
% \textbf{New York}     & 3.3507   & 3.3258   & 0.0328  \\ \hline
% \textbf{Paris}       & 6.7106   & 0.0553   & 0.0275  \\ \hline
% \textbf{Peters burg}  & 30.2450  & 48.1816  & 2.4754  \\ \hline
% \textbf{Sweden}      & 48.2092  & 174.5660 & 38.7643 \\ \hline
% \textbf{Switzerland} & 66.1216  & 150.5750 & 10.9830 \\ \hline
% \end{tabular}%
% }
% \caption{EAT Average Running times (ms)}
% \label{tab:EATRunningTime}
% \end{table}

\begin{table}[H]
\centering
% \resizebox{\columnwidth}{!}{%
\begin{tabular}{|l|c|c|c|}
\hline
& \multicolumn{3}{c|}{ \begin{tabular}[c]{@{}c@{}}\textbf{Earliest Arrival Time Algorithms} \\ Execution Time in milliseconds\end{tabular}} \\ \hline
\multicolumn{1}{|c|}{\textbf{Data Sets}} &
  \textbf{\begin{tabular}[c]{@{}c@{}} Edge-Stream \\ EAT~\cite{wu_2016_Fast_EAT_Temporal}\end{tabular}} &
  \textbf{\begin{tabular}[c]{@{}c@{}}TRG \\ EAT~\cite{Sahani2021}\end{tabular}} &
  \textbf{\begin{tabular}[c]{@{}c@{}}Our Approach \\ Algorithm~\ref{algoEarliestPathVersion2} \end{tabular} } \\ \hline
\textbf{Chicago}     & 0.79   & 1.13     & 0.07  \\ \hline
\textbf{London}      & 110.23 & 1,266.55 & 35.39 \\ \hline
\textbf{Los Angels}  & 15.73  & 192.23   & 9.73  \\ \hline
\textbf{Madrid}      & 14.51  & 240.42   & 5.42  \\ \hline
\textbf{New York}     & 3.62   & 10.67    & 0.85  \\ \hline
\textbf{Paris}       & 7.32   & 1.92     & 0.04  \\ \hline
\textbf{Petersburg}  & 32.27  & 115.73   & 6.29  \\ \hline
\textbf{Sweden}      & 51.93  & 270.18   & 31.36 \\ \hline
\textbf{Switzerland} & 70.70  & 150.34   & 15.85 \\ \hline
\end{tabular}%
% }
% \caption{EAT Execution times in milliseconds}
\caption{Run time analysis}
\label{tab:EATRunningTime}
\end{table}

% Performance and speed up in image
\begin{figure}[H]
\centering
\includegraphics[width=\columnwidth]{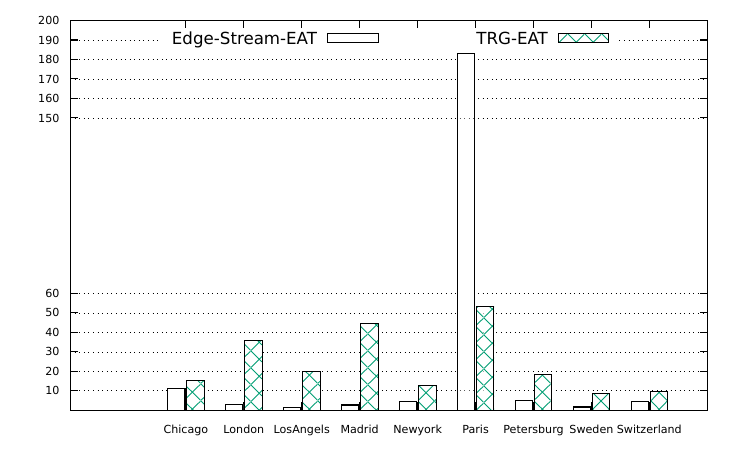}
\caption{Speed up  of Algorithm~\ref{algoEarliestPathVersion2} w.r.t the state of the art algorithms~\cite{wu_2016_Fast_EAT_Temporal,Sahani2021} }
\label{EATSpeedup}
\Description{Temporal Graph}
\end{figure}

\subsection{Performance of the Fastest Path Duration Algorithm}
\noindent
The state-of-the-art algorithms to solve \textsc{fpd} from \cite{wu_2016_Fast_EAT_Temporal} and \cite{Sahani2021} are abbreviated as \textsc{edge-stream-fpd} and \textsc{trg-fpd}, respectively. The experiment involved generating $100$ queries, each comprising $100$ source vertices selected randomly from the range $0$ to $n$, where $n$ denotes the number of vertices in the underlying graph. These same queries were employed as input for the fastest path duration Algorithm~\ref{algoFastestPathVersion2} along with two state-of-the-art fastest path duration algorithms: the edge stream-based algorithm~\cite{wu_2016_Fast_EAT_Temporal} and the \textsc{trg}-based algorithm~\cite{Sahani2021}. The evaluation was conducted on nine real world public transportation data sets as mentioned in the Table~\ref{tab: Dateset Statistics}, and the average running time for a single query was measured in milliseconds. The resulting average running times for the proposed algorithms are presented in Table \ref{tab:FPD12A}, while the speedups achieved by our approach compared to the state of the art algorithms are depicted in Fig. \ref{FPDSpeedup}.

Computing the fastest path duration using our approach, we achieved $6 \times$ average speedup over two base line algorithms \cite{wu_2016_Fast_EAT_Temporal} and \cite{Sahani2021} with maximum \textit{ $21 \times$ speedup } over algorithm \cite{wu_2016_Fast_EAT_Temporal} and \textit{$34 \times$ speedup} over algorithm proposed by \cite{Sahani2021}. 
In the public transportation application scenario, the \textsc{esd} graph data-structure is constructed once, and queried multiple times. For the data-sets shown in Table~\ref{tab: Dateset Statistics}, the pre-processing time required to construct the corresponding \textsc{esd}-graph is bounded by one minute, in practice. Hence our algorithms and implementation focus to reduce the query execution time for a given source vertex, while solving \textsc{eat} and \textsc{fpd} problems.
%, rather than analyzing the running time to construct the data-structure.

\begin{table}[H]
\centering
\resizebox{\columnwidth}{!}{
\begin{tabular}{|l|c|c|c|}
\hline
& \multicolumn{3}{c|}{ \begin{tabular}[c]{@{}c@{}}\textbf{Fastest Path Duration Algorithms} \\ Execution Time in milliseconds\end{tabular}} \\ \hline
\multicolumn{1}{|c|}{\textbf{Data Sets}} &
  \textbf{\begin{tabular}[c]{@{}c@{}} \textsc{edge-stream} \\ \textsc{fpd} ~\cite{wu_2016_Fast_EAT_Temporal}\end{tabular}} &
  \textbf{\begin{tabular}[c]{@{}c@{}}\textsc{trg} \\ \textsc{fpd}~\cite{Sahani2021}\end{tabular}} &
  \textbf{\begin{tabular}[c]{@{}c@{}}Our Approach \\ Algorithm~\ref{algoFastestPathVersion2} \end{tabular} } \\ \hline
\textbf{Chicago}     & \multicolumn{1}{c|}{1.90} & \multicolumn{1}{c|}{1.83} & 0.52 \\ \hline
\textbf{London}      & \multicolumn{1}{c|}{6885.24} & \multicolumn{1}{c|}{2880.80}  & 1576.75 \\ \hline
\textbf{Los Angels}  & \multicolumn{1}{c|}{874.91} & \multicolumn{1}{c|}{481.14} & 179.56 \\ \hline
\textbf{Madrid}      & \multicolumn{1}{c|}{1756.52} & \multicolumn{1}{c|}{530.86} & 212.09 \\ \hline
\textbf{Newyork}     & \multicolumn{1}{c|}{91.72} & \multicolumn{1}{c|}{18.81} & 4.44 \\ \hline
\textbf{Paris}       & \multicolumn{1}{c|}{1.67} & \multicolumn{1}{c|}{7.77} & 0.23 \\ \hline
\textbf{Petersburg}  & \multicolumn{1}{c|}{1305.77} & \multicolumn{1}{c|}{396.88} & 198.10 \\ \hline
\textbf{Sweden}      & \multicolumn{1}{c|}{1346.45} & \multicolumn{1}{c|}{975.65} & 469.56   \\ \hline
\textbf{Switzerland} & \multicolumn{1}{c|}{652.72} & \multicolumn{1}{c|}{489.88} & 182.03 \\ \hline
\end{tabular}
}
\caption{Run time analysis}
\label{tab:FPD12A}
\end{table}

\begin{figure}[H]
\centering
\includegraphics[width=\columnwidth]{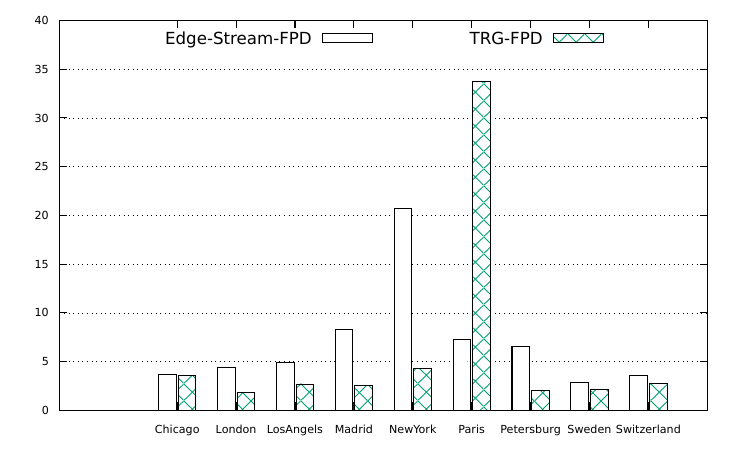}
\caption{Speed up  of Algorithm~\ref{algoFastestPathVersion2} w.r.t the state of the art algorithms~\cite{wu_2016_Fast_EAT_Temporal,Sahani2021} }
\label{FPDSpeedup}
\Description{Temporal Graph}
\end{figure}

\subsection{Key Insights}

Our earliest arrival time algorithm process at most $2\%$ of the edges and fastest path duration algorithm process at most $70\%$ of the edges, as shown in Figure~\ref{Fig:EAT_conn_processed} and Figure ~\ref{Fig:FPD_conn_processed}.  This is the key reason to the beat running time of the existing algorithms in practice.

\begin{figure}[H]
\centering
\includegraphics[width=\columnwidth]{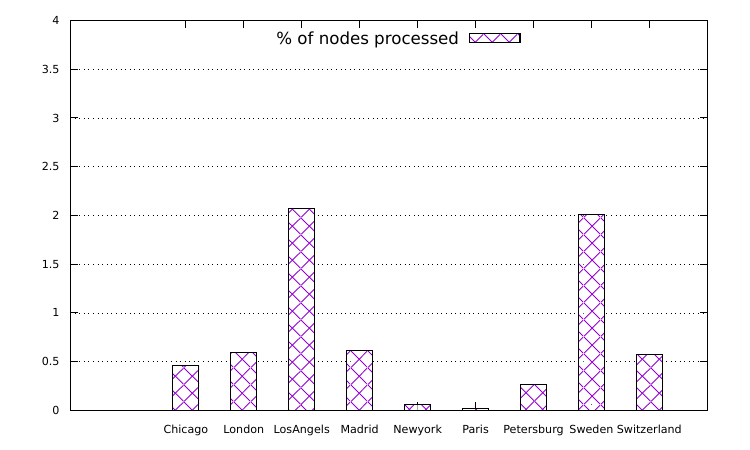}
\caption{\% of nodes processed by Algorithm~\ref{algoEarliestPathVersion2} }
\label{Fig:EAT_conn_processed}
\Description{Temporal Graph}
\end{figure}

\begin{figure}[H]
\centering
\includegraphics[width=\columnwidth]{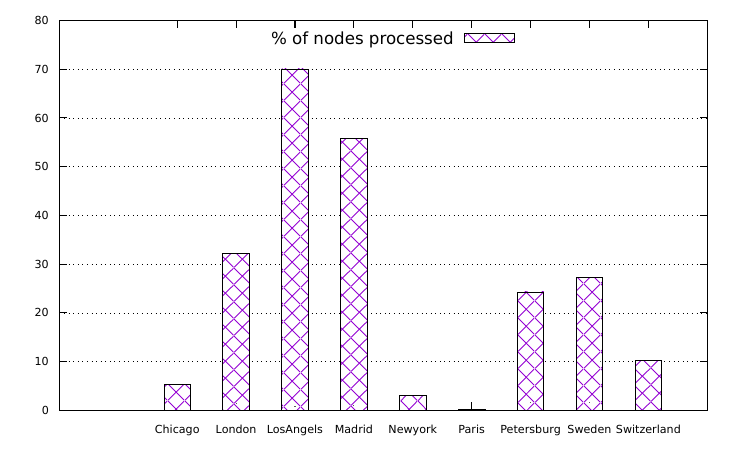}
\caption{\% of nodes processed by Algorithm~\ref{algoFastestPathVersion2} }
\label{Fig:FPD_conn_processed}
\Description{Temporal Graph}
\end{figure}

% \begin{table}[H]
% \resizebox{\columnwidth}{!}{%
% \begin{tabular}{|l|c|cc|}
% \hline
% \multirow{2}{*}{Data Sets} &
%   \multirow{2}{*}{|E(G)|} &
%   \multicolumn{2}{c|}{\textbf{Average number of nodes processed}} \\ \cline{3-4} 
%  &
%    &
%   \multicolumn{1}{c|}{\textbf{\begin{tabular}[c]{@{}c@{}}EAT \\ Algorithm~\ref{algoEarliestPathVersion2}\end{tabular}}} &
%   \textbf{\begin{tabular}[c]{@{}c@{}}FPD \\ Algorithm~\ref{algoFastestPathVersion2}\end{tabular}} \\ \hline
% \textbf{Chicago}     & 98157    & \multicolumn{1}{c|}{456}    & 5205    \\ \hline
% \textbf{London}      & 14064967 & \multicolumn{1}{c|}{82976}  & 4530119 \\ \hline
% \textbf{Los Angels}  & 1979340  & \multicolumn{1}{c|}{41069}  & 1386118 \\ \hline
% \textbf{Madrid}      & 1994688  & \multicolumn{1}{c|}{12238}  & 1112698 \\ \hline
% \textbf{Newyork}     & 514390   & \multicolumn{1}{c|}{326}    & 15151   \\ \hline
% \textbf{Paris}       & 1068284  & \multicolumn{1}{c|}{210}    & 1552    \\ \hline
% \textbf{Petersburg}  & 4437010  & \multicolumn{1}{c|}{11430}  & 1075480 \\ \hline
% \textbf{Sweden}      & 6567745  & \multicolumn{1}{c|}{132250} & 1789260 \\ \hline
% \textbf{Switzerland} & 9261315  & \multicolumn{1}{c|}{52832}  & 939813  \\ \hline
% \end{tabular}%
% }
% \caption{Average number of nodes in $E(G)$ processed in edge-scan-dependency-graph}
% \label{tab:Connections_Processed}
% \end{table}

\begin{table}[H]
\centering
\resizebox{\columnwidth}{!}{%
\begin{tabular}{|l|lc|}
\hline
\multicolumn{1}{|c|}{\multirow{2}{*}{\textbf{Data Sets}}} & \multicolumn{2}{c|}{\textbf{Our Approach Algorithm~\ref{algoEarliestPathVersion2}}}                        \\ \cline{2-3} 
\multicolumn{1}{|c|}{}                                    & \multicolumn{1}{c|}{\textbf{Running Time (ms)}} & \textbf{\# Nodes Processed} \\ \hline
\textbf{Paris}             & \multicolumn{1}{c|}{0.04}                       & 210                         \\ \hline
\textbf{Newyork}           & \multicolumn{1}{c|}{0.85}                       & 326                         \\ \hline
\textbf{Chicago}           & \multicolumn{1}{c|}{0.07}                       & 456                         \\ \hline
\textbf{Petersburg}        & \multicolumn{1}{c|}{6.29}                       & 11430                       \\ \hline
\textbf{Madrid}            & \multicolumn{1}{c|}{5.42}                       & 12238                       \\ \hline
\textbf{Los Angels}        & \multicolumn{1}{c|}{9.73}                       & 41069                       \\ \hline
\textbf{Switzerland}       & \multicolumn{1}{c|}{15.85}                      & 52832                       \\ \hline
\textbf{London}            & \multicolumn{1}{c|}{35.39}                      & 82976                       \\ \hline
\textbf{Sweden}            & \multicolumn{1}{c|}{31.36}                      & 132250                      \\ \hline
\end{tabular}%
}
\caption{Average number of nodes in $E(G)$ processed in edge-scan-dependency-graph}
\label{tab:EAT_AVG_Conn}
\end{table}

\begin{table}[H]
\centering
\resizebox{\columnwidth}{!}{%
\begin{tabular}{|l|cc|}
\hline
\multirow{2}{*}{Data Sets} & \multicolumn{2}{c|}{\textbf{Our Approach Algorithm~\ref{algoFastestPathVersion2}}}     \\ \cline{2-3} 
                           & \multicolumn{1}{c|}{\textbf{Running Time (ms)}} & \textbf{\# Nodes Processed} \\ \hline
\textbf{Paris}             & \multicolumn{1}{c|}{0.23}                     & 1517                        \\ \hline
\textbf{Chicago}           & \multicolumn{1}{c|}{0.52}                    & 5931                        \\ \hline
\textbf{Newyork}           & \multicolumn{1}{c|}{4.44}                    & 19779                       \\ \hline
\textbf{Switzerland}       & \multicolumn{1}{c|}{182.03}                    & 1127507                     \\ \hline
\textbf{Los Angels}        & \multicolumn{1}{c|}{179.56}                     & 1308683                     \\ \hline
\textbf{Petersburg}        & \multicolumn{1}{c|}{198.1}                     & 1615717                     \\ \hline
\textbf{Madrid}            & \multicolumn{1}{c|}{212.09}                    & 1655109                     \\ \hline
\textbf{Sweden}            & \multicolumn{1}{c|}{469.56}                    & 2787875                     \\ \hline
\textbf{London}            & \multicolumn{1}{c|}{1576.75}                    & 6232547                     \\ \hline
\end{tabular}%
}
\caption{Average number of nodes in $E(G)$ processed in edge-scan-dependency-graph}
\label{tab:FPD_AVG_Conn}
\end{table}

The percentage of edges of $G$ (nodes of $\Tilde{G}$) processed by  Algorithm~\ref{algoEarliestPathVersion2} and Algorithm~\ref{algoFastestPathVersion2} are shown in Figure~\ref{Fig:EAT_conn_processed} and Figure~\ref{Fig:FPD_conn_processed}, respectively. For most of the data sets, the execution times of Algorithm~\ref{algoEarliestPathVersion2} and Algorithm~\ref{algoFastestPathVersion2} are proportional to the number of nodes processed in the respective algorithms. For instance, the execution time of Algorithm~\ref{algoEarliestPathVersion2} on data sets Sweden and London is high, as the number of nodes processed is higher. Similarly, the running time of Algorithm~\ref{algoFastestPathVersion2} on data-sets Sweden, Switzerland, and London are high, because the number of nodes being explored are higher.
These insights can be observed from Table~\ref{tab:EAT_AVG_Conn} and Table~\ref{tab:FPD_AVG_Conn}.

% \begin{table}[H]
% \resizebox{\columnwidth}{!}{%
% \begin{tabular}{|l|c|cc|}
% \hline
% \multirow{2}{*}{Data Sets} &
%   \multirow{2}{*}{|E(G)|} &
%   \multicolumn{2}{c|}{\textbf{Average number of nodes processed}} \\ \cline{3-4} 
%  &
%    &
%   \multicolumn{1}{c|}{\textbf{\begin{tabular}[c]{@{}c@{}}EAT \\ Algorithm~\ref{algoEarliestPathVersion2}\end{tabular}}} &
%   \textbf{\begin{tabular}[c]{@{}c@{}}FPD \\ Algorithm~\ref{algoFastestPathVersion2}\end{tabular}} \\ \hline
% \textbf{Chicago}     & 98157    & \multicolumn{1}{c|}{456}    & 5205    \\ \hline
% \textbf{London}      & 14064967 & \multicolumn{1}{c|}{82976}  & 4530119 \\ \hline
% \textbf{Los Angels}  & 1979340  & \multicolumn{1}{c|}{41069}  & 1386118 \\ \hline
% \textbf{Madrid}      & 1994688  & \multicolumn{1}{c|}{12238}  & 1112698 \\ \hline
% \textbf{Newyork}     & 514390   & \multicolumn{1}{c|}{326}    & 15151   \\ \hline
% \textbf{Paris}       & 1068284  & \multicolumn{1}{c|}{210}    & 1552    \\ \hline
% \textbf{Petersburg}  & 4437010  & \multicolumn{1}{c|}{11430}  & 1075480 \\ \hline
% \textbf{Sweden}      & 6567745  & \multicolumn{1}{c|}{132250} & 1789260 \\ \hline
% \textbf{Switzerland} & 9261315  & \multicolumn{1}{c|}{52832}  & 939813  \\ \hline
% \end{tabular}%
% }
% \caption{Average number of nodes in $E(G)$ processed in edge-scan-dependency-graph}
% \label{tab:Connections_Processed}
% \end{table}

Many public transport administrators explore various insights on their cities and suggest the public to start a journey from a suitable time, that minimize the duration time. From our experiments, on many data-sets, we observe that the suitable time to start a journey is around $6$ AM or $3$ PM, to minimize the journey time. 

\begin{figure}[H]
\centering
\begin{subfigure}[b]{0.4\textwidth}
\centering
\includegraphics[width=\textwidth]{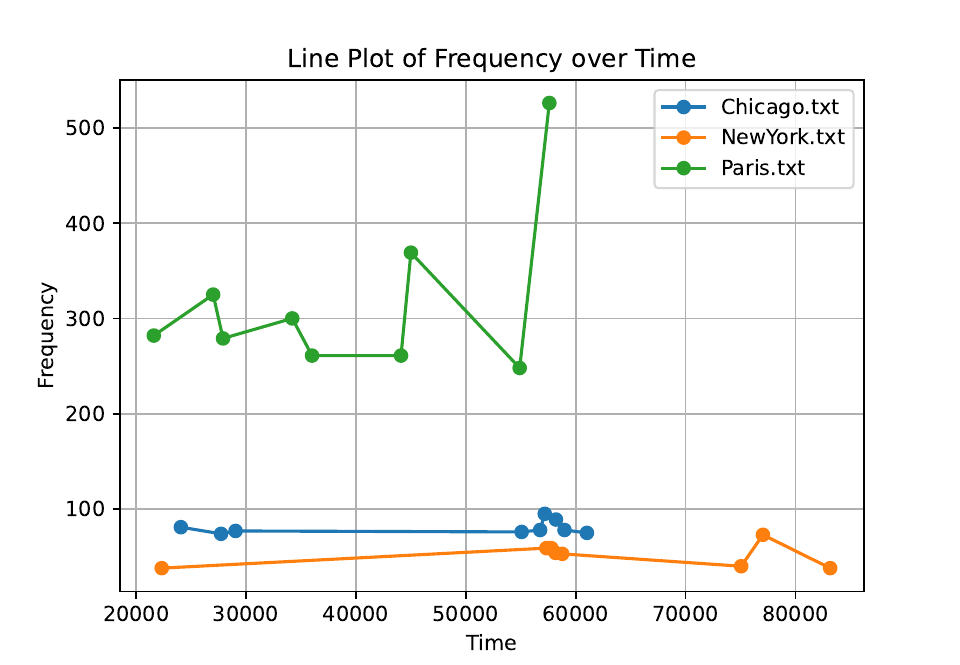}
% \caption{Frequencies of starting times}
% \label{fig:freq1}
\end{subfigure}
\hfill
\begin{subfigure}[b]{0.4\textwidth}
\centering
\includegraphics[width=\textwidth]{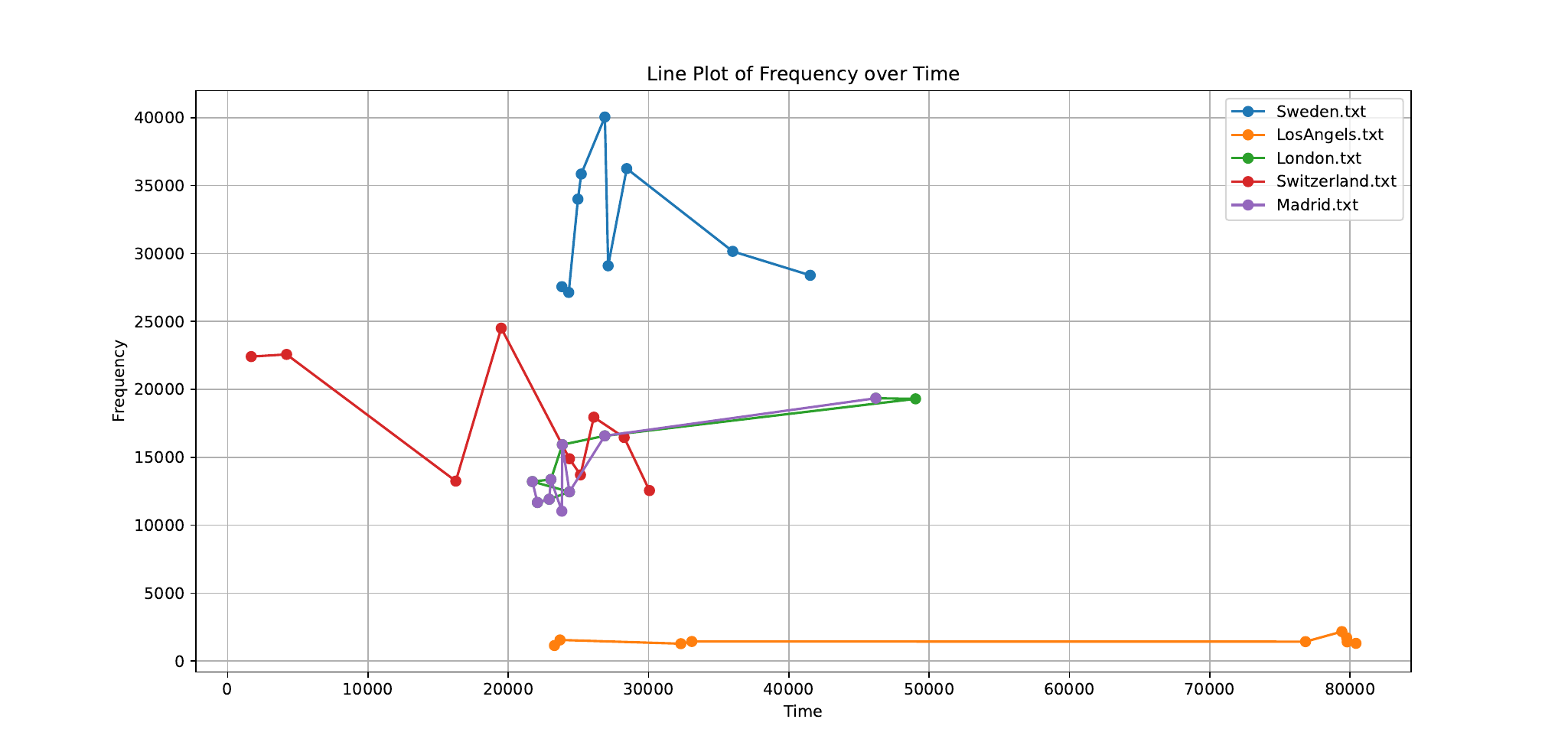}
% \caption{Frequencies of starting times}
% \label{fig:freq2}
\end{subfigure}
\caption{Frequencies of starting times}
\label{fig:freq1}
\Description{Temporal Graph}
\end{figure}

% These insights can be observed from Figure~\ref{fig:freq1} and Figure~\ref{fig:freq2}

% \begin{figure}[H]
% \centering
% \includegraphics[width=\columnwidth]{Images/Freq_2.pdf}
% \caption{Frequencies of starting times}
% \label{fig:freq1}
% \end{figure}

% \begin{figure}[H]
% \centering
% \includegraphics[width=\columnwidth]{Images/Freq_1.pdf}
% \caption{Frequencies of starting times}
% \label{fig:freq2}
% \end{figure} 

% We are conducting additional experiments with 100 generated queries comprising source vertices and various readiness times. We define a range of readiness times, including starting, middle, and maximum times. These times are applied to the source vertex to create different scenarios, and we compute the earliest arrival time using Algorithm \ref{algoEarliestPathVersion1}. The results are displayed in Table.

% it shows d and analyzed to assess how the earliest arrival time varies with different start times, enabling insights into the efficiency and performance of the algorithms in the context of temporal graph analysis.
% \newpage
\section{Applications}

Our Algorithms helps to solve various path problems in various domains that help public transport administrator.

\begin{itemize}[left=0pt, itemsep=0.5pt]
    \item \textit{How much time is required to visit any place in the city from the given source location?}

    This question is addressed by calculating the fastest path duration to reach the farthest or longest location from the provided source location. This information is beneficial for public transport administrators in determining the time needed to reach any place in the city from a given source location. It aids in planning efficient routes and optimizing travel time.
    % This question can be answered by computing the fastest path duration required to reach farthest or longest location from the given source location. It will be helpful for public transport administrators to determine the time required to visit any place in the city from a given source location. This helps in planning efficient routes and optimizing travel time. 

    \item \textit{Coverage Analysis - How many places can be covered in $k$ hours from the given source location?}
    % \textit{Coverage Analysis - Determining the Reach within $k$ Hours from the Source Location}
    
    Administrators can evaluate the number of places that can be covered within a designated time frame. For instance, they may want to ascertain how many locations can be visited within $k$ hours from a specified source point. This analysis proves valuable for resource allocation and optimizing service coverage. This challenge is addressed by incorporating a modification in Algorithm \ref{algoFastestPathVersion2}, terminating each phase if the duration surpasses $k$ hours.  
    % Administrators can assess how many places can be covered within a specific time constraint, such as determining the number of places that can be visited within $k$ hours from the given source location. This aids in resource allocation and optimizing service coverage. This problem can be solved by terminating each phase of the fastest path algorithm \ref{algoFastestPathVersion1} if duration exceeds $k$ hours.

    \item \textit{Percentage Coverage - How much time is required to cover $k\%$ of the city from the given source location?}
    % \textit{Percentage Coverage - Determining the time required to cover $k\%$ of the city from the given source location.}

    Administrators can ascertain the time needed to cover a specific percentage of the city, such as determining the duration required to cover $k\%$ of the city from the provided source location. This insight aids in understanding the reach and accessibility of services across different areas. Addressing this challenge involves incorporating a termination condition in each phase of the Algorithm ~\ref{algoFastestPathVersion1}. Termination occurs when the coverage of a specific phase surpasses $k\%$. The algorithm monitors the maximum journey duration in each phase, and identify the smallest one, known as the global min-journey time. A local phase can be terminated if the local journey time exceeds either the global min-journey time or $k\%$ coverage in the current phase. Additionally, the global min journey time is updated whenever a better local journey time is achieved in the current phase.
    % Administrators can determine the time required to cover a certain percentage of the city, such as finding the duration needed to cover 65\% of the city from the given source location. This provides insights into the reach and accessibility of services across different areas. This problem can be addressed by introducing a termination condition in each phase of the fastest path algorithm \ref{algoFastestPathVersion1}. The termination occurs when the coverage of a specific phase exceeds $k\%$. The algorithm keeps track of the maximum journey duration in each phase, referred to as the global min-journey time. A local phase can be terminated if the local journey time exceeds either the global min-journey time or 65\% coverage in the current phase. Additionally, the global min journey time is updated whenever a better local journey time is obtained in the current phase.
    
\end{itemize}

\section{Related Work}

The goal-oriented variants of various problems are extensively explored in the context of public transport networks. In particular, multiple algorithms such as \textsc{raptor}, transfer patterns, connection scan accelerated, and trip-based have been designed to extract paths ranging from earliest arrival and profile search to multi-criteria paths ~\cite{raptor_delling2015,transferPatterns_2010_Hannah,connectionScanAcc_2014_Ben,dibbelt2018connectionScan,tripBased_2015_Sascha}. When dealing with real world temporal graphs, the goal-oriented fastest paths can be retrieved using indexing techniques \textsc{ttl} and \textsc{top-chain} ~\cite{ttlSIGMOD2015,topChainICDE2016_wu_rechablity}. All these algorithms are primarily targeted for goal-oriented and not for single-source variant.

Our focus is now on single-source variants of \textsc{eat} and \textsc{fpd}. Xuan et.al have designed a vertex centric algorithm to solve single-source \textsc{eat} problem ~\cite{foremostJourney2003}. Later, this is improved using edge stream representation and the associated edge centric algorithm ~\cite{wu_2016_Fast_EAT_Temporal}. Similarly, \textsc{fpd} is targeted using edge centric ~\cite{wu_2016_Fast_EAT_Temporal} and vertex centric algorithms on the transformed graphs. In particular, many graph transformations ~\cite{wu_2016_Fast_EAT_Temporal,zschoche2020_TRG,Sahani2021} are developed to solve single-source shortest and fastest path problems. In all these transformations, a temporal graph is transformed to an equivalent time respected graph, in which the departure and arrival times of temporal edges are treated as vertices and added edges to capture the necessary dependencies. The key idea in these works is to reduce the number of vertices and edges in the transformed graphs. In this paper, we have compared our results with edge-stream and \textsc{trg} algorithms, which are the state-of-the-art techniques to solve \textsc{eat} and \textsc{fpd} problems.   

Due to the wide range of applications, building efficient solutions on public transport networks has received attention from both traditional and machine learning algorithms \cite{2023_Liu_min_transfers,2020_BigData_Applications_ML,2020_ICCE_Panovski_prediction_ML}.
%\textcolor{red}{few more lines to be written on recent works or interesting works on public transport problems.. few lines will be enough}
Recent research on public transportation encompasses diverse areas. Letelier et al. (2023) focused on compacting large public transport data ~\cite{letelier2023compacting}. Dahlmanns et al. (2023) optimized transportation networks considering congestion \cite{dahlmanns2023optimizing}, while Cao et al. (2023) sought to enhance public transportation quality through dynamic bus departure times \cite{cao2023improving}. Drabicki et al. highlights the significant impact of capacity-constrained models on transportation outcomes \cite{drabicki2023public}. These studies contribute to data efficiency, network optimization, and service quality improvements.
Machine learning algorithms are used to find important features in the public transport data, influencing whether a vehicle stops on time or late \cite{2020_BigData_Applications_ML}. Also, the prediction of arrival times of public transport vehicles is well studied and useful in daily routine \cite{2020_ICCE_Panovski_prediction_ML}.

\section{Conclusion}
% Our research aimed to tackle fundamental path problems within a public transportation network represented as a temporal graph. We achieved substantial efficiency improvements through a comprehensive analysis and comparisons with state-of-the-art algorithms. Notably, the fastest path problem exhibited a remarkable 35-fold speedup, while the earliest arrival time problem showed an even more impressive 101-fold acceleration. Our algorithms and data-structures have potential to solve other variants of path problems and help to perform public transport analytics.
In this research, we addressed key path finding challenges in public transportation networks by developing efficient near linear-time algorithms for the earliest arrival time and fastest path duration problems. Utilizing edge-scan-dependency graphs, our approach significantly outperformed existing methods, evidenced by a 34-fold speedup in \textsc{fpd} and an unprecedented 183-fold improvement in \textsc{eat}. Key to our success was the novel use of useful dominating paths and edge-scan-dependency graph data structures, ensuring minimal edge processing. These advancements were empirically validated on real-world datasets, demonstrating not only theoretical innovation but also practical effectiveness in urban transit systems. Our work marks a significant step forward in optimizing public transportation routes, offering powerful tools for transit planners and setting a new standard in algorithmic solutions for urban mobility challenges.

\bibliographystyle{ACM-Reference-Format}
	\bibliography{main.bib}

\end{document}